\documentclass[twoside,11pt]{article}

\usepackage[preprint, abbrvbib]{jmlr2e}
\usepackage{url}            
\usepackage{booktabs}       
\usepackage{amsfonts}       
\usepackage{nicefrac}       
\usepackage{microtype}      
\usepackage{natbib}
\usepackage{subfigure}
\usepackage{booktabs}

\usepackage{eqnarray,amsmath}
\usepackage{algorithm}
\usepackage{algorithmic}
\usepackage{graphicx}
\usepackage{multirow}
\usepackage{longtable}
\usepackage[hang]{caption2}
\usepackage[shortlabels]{enumitem}

\newcommand{\E}{\mathbb{E}}
\renewcommand{\Pr}{\mathsf{P}}

\newcommand{\dGam}{\text{Gamma}}
\newcommand{\dNorm}{\text{N}}
\newcommand{\dUnif}{\text{U}}
\newcommand{\dExp}{\text{Exp}}
\newcommand{\dFNorm}{\text{FN}}
\newcommand{\cd}{~\vert~}
\newcommand{\cdotmid}{\,\cdot \mid}

\newcommand{\Real}{\mathbb{R}}
\newcommand{\dd}{\mathrm{d}}

\newcommand{\push}[1]{#1_{\sharp}}

\newcommand{\ie}{i.e.,}
\newcommand{\eg}{e.g.,}

\newcommand{\probfam}{\mathcal{P}}
\newcommand{\kernfam}{\mathcal{K}}

\newcommand{\simiid}{\overset{\text{iid}}{\sim}}
\newcommand{\distarrow}{\overset{d}{\rightarrow}}
\newcommand{\asarrow}{\overset{a.s.}{\rightarrow}}

\usepackage{xcolor}

\newcommand{\sy}{\tilde{y}}

\DeclareMathOperator*{\argmax}{arg\,max}

\newcommand{\ArgMax}[1]{\raisebox{0.5ex}{\scalebox{0.8}{$\displaystyle \argmax_{#1}\;$}}}

\usepackage{tcolorbox}

\usepackage{lastpage}
\jmlrheading{27}{2026}{1--\pageref{LastPage}}{7/24}{1/26}{24-1179}{Joshua J. Bon, David J. Warne, David J. Nott, Christopher Drovandi}

\ShortHeadings{Bayesian Score Calibration}{Bon, Warne, Nott, and Drovandi}
\firstpageno{1}

\begin{document}

\title{Bayesian Score Calibration for Approximate Models}

\author{\name Joshua J. Bon\thanks{Thanks to Ming Xu, Aad van der Vaart, and anonymous referees for their helpful comments.} \email joshua.bon@adelaide.edu.au \\
\addr School of Mathematical Sciences,\\ Adelaide University
       \AND
\name David J. Warne \email david.warne@qut.edu.au \\
\addr School of Mathematical Sciences \& Centre for Data Science, \\ Queensland University of Technology \\
ARC Centre of Excellence for the Mathematical Analysis of Cellular Systems
        \AND
\name David J. Nott \email standj@nus.edu.sg \\
\addr Department of Statistics and Data Science,\\
  National University of Singapore
        \AND
\name Christopher Drovandi \email c.drovandi@qut.edu.au \\
\addr School of Mathematical Sciences \& Centre for Data Science, \\ Queensland University of Technology \\
ARC Centre of Excellence for the Mathematical Analysis of Cellular Systems
}

\editor{Edo Airoldi}

\maketitle

\begin{abstract}
Scientists continue to develop increasingly complex mechanistic models to reflect their knowledge more realistically. Statistical inference using these models can be challenging since the corresponding likelihood function is often intractable and model simulation may be computationally burdensome.  Fortunately, in many of these situations it is possible to adopt a surrogate model or approximate likelihood function.  It may be convenient to conduct Bayesian inference directly with a surrogate, but this can result in a posterior with poor uncertainty quantification.  In this paper, we propose a new method for adjusting approximate posterior samples to reduce bias and improve posterior coverage properties.  We do this by optimizing a transformation of the approximate posterior, the result of which maximizes a scoring rule.  Our approach requires only a (fixed) small number of complex model simulations and is numerically stable.  We develop supporting theory for our method and demonstrate beneficial corrections to approximate posteriors across several examples of increasing complexity. 
\end{abstract}


\section{Introduction}

Scientists and practitioners desire greater realism and complexity in their models, but this can complicate likelihood-based inference.  If the proposed model is sufficiently complex, computation of the likelihood can be intractable.  In this setting, if model simulation is feasible, then approximate Bayesian inference can proceed via likelihood-free methods \citep{Sisson2018}.  However, most likelihood-free methods require a large number of model simulations. It is common for likelihood-free inference methods to require hundreds of thousands of model simulations or more.  Thus, if model simulation is also computationally intensive, it is difficult to conduct inference via likelihood-free methods.

Often it is feasible to propose an approximate but more computationally tractable ``surrogate'' version of the model of interest. However, the resulting approximate Bayesian inferences can be biased (relative to the true posterior), and produce distributions with poor uncertainty quantification that do not have correct coverage properties \citep[see for example,][]{pmlr-v97-xing19a,Warne2022}.

In this paper, we propose a novel Bayesian procedure for approximate inference. Related literature will be reviewed in Section~\ref{sec:related}. Our approach begins with an approximate model and transforms the resulting approximate posterior to reduce bias and more accurately quantify uncertainty.  Only a small number (\ie\ hundreds) of model simulations from the target model are required, whilst standard Bayesian inference is only conducted using the approximate model.  Furthermore, these computations are trivial to parallelize, which can reduce the time cost further by an order of magnitude or more, depending on available computing resources.  Importantly, our procedure does not require any evaluations of the likelihood function of the complex model of interest.  

The approximate posterior can be formed on the basis of a surrogate model or approximate likelihood function. For example, a surrogate model may be a deterministic version of a complex stochastic model \citep[\eg][]{Warne2022} and an example of a surrogate likelihood is the Whittle likelihood for time series models \citep{whittle1953estimation}.  Furthermore, our framework permits the application of approximate Bayesian inference algorithms on the surrogate model. For example, the Laplace approximation, variational approximations, and likelihood-free inference methods that require only a small number of model simulations \citep[e.g.,][]{gutmann2016bayesian}. Hence, computational approximations and surrogate models can be used, and corrected, simultaneously.
\begin{figure}
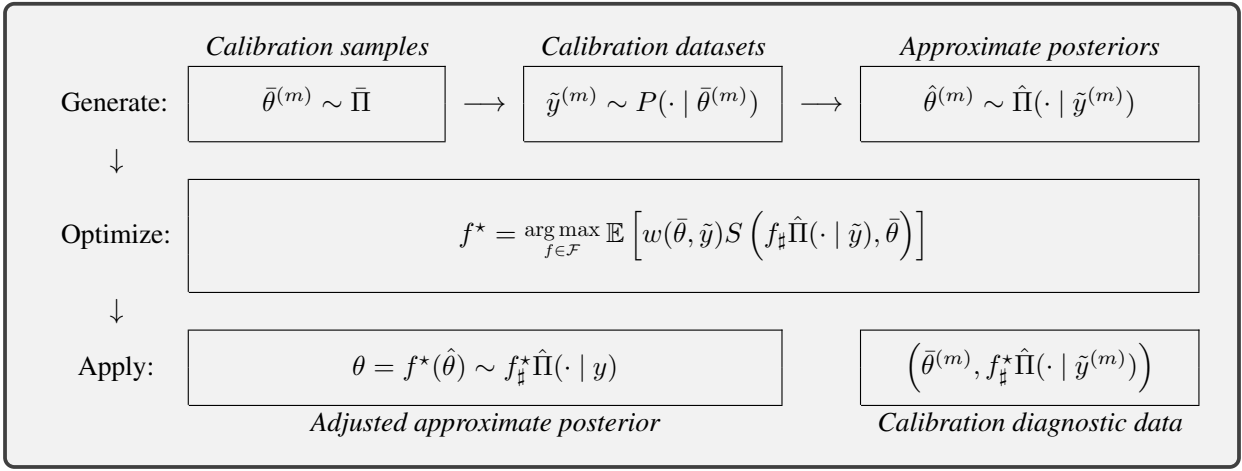

\centering
\tcbox{
\resizebox{0.93\textwidth}{!}{%
\begin{tabular}{cccccc}
    & \textit{Calibration samples} &     & \textit{Calibration data sets} &    & \textit{Approximate posteriors} \\ 
\cline{2-2} \cline{4-4} \cline{6-6} 
\multicolumn{1}{l|}{\multirow{2}{*}{Generate:}} & \multicolumn{1}{c|}{\multirow{2}{*}{$\bar{\theta}^{(m)} \sim \bar{\Pi}$}} & \multicolumn{1}{l|}{\multirow{2}{*}{$\longrightarrow$}} & \multicolumn{1}{c|}{\multirow{2}{*}{$\sy^{(m)} \sim P(\cdotmid \bar{\theta}^{(m)})$}} & \multicolumn{1}{l|}{\multirow{2}{*}{$\longrightarrow$}} & \multicolumn{1}{c|}{\multirow{2}{*}{$\hat{\theta}^{(m)} \sim \hat{\Pi}(\cdotmid \sy^{(m)})$}} \\
\multicolumn{1}{l|}{}                             & \multicolumn{1}{l|}{}                  & \multicolumn{1}{l|}{}                  & \multicolumn{1}{l|}{}                  & \multicolumn{1}{l|}{}                  & \multicolumn{1}{l|}{} \\ \cline{2-2} \cline{4-4} \cline{6-6} 
$\downarrow$ & & & & & \\ \cline{2-6} 
\multicolumn{1}{l|}{\multirow{3}{*}{Optimize:}} & \multicolumn{5}{c|}{\multirow{3}{*}{$f^{\star} = \ArgMax{f \in \mathcal{F}} \E\left[w(\bar{\theta}, \sy) S \left(\push{f}\hat{\Pi}(\cdotmid \sy),\bar{\theta}\right) \right]$}} \\ 
\multicolumn{1}{l|}{} & \multicolumn{5}{l|}{} \\
\multicolumn{1}{l|}{} & \multicolumn{5}{l|}{} \\ \cline{2-6}
$\downarrow$ & & & & & \\ \cline{2-4} \cline{6-6} 
\multicolumn{1}{c|}{\multirow{2}{*}{Apply:}} & \multicolumn{3}{c|}{\multirow{2}{*}{$\theta = f^{\star}(\hat\theta) \sim \push{f}^{\star}\hat{\Pi}(\cdotmid y)$}} & \multicolumn{1}{c|}{\multirow{2}{*}{}} & \multicolumn{1}{c|}{\multirow{2}{*}{$\left(\bar{\theta}^{(m)}, \push{f}^\star \hat{\Pi}(\cdotmid \sy^{(m)}) \right)$}} \\
\multicolumn{1}{l|}{} & \multicolumn{3}{l|}{} & \multicolumn{1}{l|}{} & \multicolumn{1}{l|}{} \\ \cline{2-4} \cline{6-6}
\multicolumn{1}{l}{} & \multicolumn{3}{c}{\textit{Adjusted approximate posterior}} & \multicolumn{1}{l}{} & \multicolumn{1}{c}{\textit{Calibration diagnostic data}} \\
\end{tabular}%
}} \caption{Graphical overview of Bayesian score calibration. Firstly, the importance distribution $\bar{\Pi}$ and data-generating process $P(\cdotmid \theta)$ simulate parameter-data pairs $(\bar{\theta}^{(m)}, \sy^{(m)})$ for $m \in \{1,\ldots,M\}$. Each simulated data set, $\sy^{(m)}$, defines a new approximate posterior, $\hat{\Pi}(\cdotmid \sy^{(m)})$, which we approximate with Monte Carlo samples. Secondly, we use a strictly proper scoring rule $S$ to find the best transformation of the approximate posterior, defining the pushforward distribution $\push{f}\hat{\Pi}(\cdotmid \sy)$, with respect to true data-generating parameter $\bar\theta$, averaged over~$\bar\Pi$, with weights $w(\bar\theta, \sy)$. The optimization objective function is approximated using Monte Carlo with pre-computed samples from the generation step. Finally, the optimal function, $f^\star$, is used to generate samples from the adjusted approximate posterior with the observed data, $y$, and produce data for diagnostic summaries.}\label{schematic-bsc}
\end{figure}

Figure~\ref{schematic-bsc} displays a graphical overview of Bayesian score calibration, which we describe in detail throughout the paper. We develop a new theoretical framework to support our method that is also applicable to some related methods. In particular, Theorem~\ref{th:SBIresult} generalizes the underlying theory justifying recent simulation-based inference methods  \citep[\eg][]{pacchiardi2022likelihood}. Moreover, we provide practical calibration diagnostics that assess the quality of the adjusted approximate posteriors and alert users to success or failure of the procedure.

The rest of this paper is organized as follows. In Section \ref{sec:methods} we provide background, present our approximate model calibration method, and develop theoretical justifications.  Our new method is demonstrated on some examples of increasing complexity in Section \ref{sec:examples} with further examples deferred to the supplementary materials \ref{sec:ex-gauss}--\ref{sec:lotka-abc}. Section \ref{sec:discussion} contains a concluding discussion of limitations, as well as ongoing and future research directions.

\section{Methods} \label{sec:methods}


Bayesian inference updates the prior distribution of unknown parameters~${\theta \sim \Pi}$ with information from observed data $y \sim P(\cdotmid \theta)$ having some dependence on $\theta$. We take~$\Pi$ as the prior probability measure with probability density (or mass) function $\pi$, and $P(\cdotmid \theta)$ as the data-generating process with probability density (or mass) function $p(\cdotmid \theta)$. We will use the terms distribution and law of a random variable with reference to an underlying probability measure. The function $p(y \cd \cdot\,)$ is the likelihood for fixed data $y$ and varying $\theta$. The prior distribution~$\Pi$ is defined on measurable space $(\Theta,\vartheta)$, whilst $P(\cdotmid \theta)$ is defined on $(\mathsf{Y},\mathcal{Y})$ for fixed $\theta \in \Theta$, where $\vartheta$ (resp. $\mathcal{Y}$) is a $\sigma$-algebra on $\Theta$ (resp. $\mathsf{Y}$). 

Bayes theorem determines that the posterior distribution, incorporating information from the data, has density (mass) function $\pi(\theta \cd y) = \frac{p(y \cd \theta) \pi(\theta)}{p(y)}$,
where $p(y) = \int p(y \cd \theta) \Pi(\dd \theta)$ for $\theta \in \Theta$. Here, for fixed $y \in \mathsf{Y}$, $Z = p(y)$ is the posterior normalizing constant. For varying~$y \in \mathsf{Y}$, $p(y)$ is the density (mass) function of the marginal distribution of the data~$P$, defined on $(\mathsf{Y},\mathcal{Y})$.

Posterior inference typically requires approximate methods as the normalizing constant~$Z$ is unavailable in a closed form. There are two broad families of approximations in general use: (i) sampling methods, including Markov chain Monte Carlo \citep[MCMC,][]{brooks2011handbook} and sequential Monte Carlo \citep[SMC,][]{chopin2020introduction}; and (ii) optimization-based methods including variational inference or Laplace approximations \citep{bishop2006pattern}. All standard implementations of these methods rely on pointwise evaluation of the likelihood function, or some unbiased estimate.

When the likelihood is infeasible or computationally expensive to evaluate pointwise we may wish to choose a surrogate posterior which approximates the original in some sense. We consider a surrogate with distribution $\hat{\Pi}(\cdotmid y )$ on~$(\Theta,\vartheta)$ and with density (mass) function~$\hat{\pi}(\cdotmid y )$. Such a surrogate can arise from an approximation to the original model, likelihood or posterior. In the case of a surrogate model or approximate likelihood, this is equivalent to $\hat{\pi}(\theta \cd y ) \propto \hat{p}( y \cd \theta)\pi(\theta)$ for an approximate likelihood $\hat{p}( y \cd \theta)$.

We design methods to calibrate the approximate posterior when the true likelihood,~$p(y \cd \theta)$, cannot be evaluated but we can sample from the data-generating process~$P(\cdotmid \theta)$. To facilitate this calibration we need a method to compare approximate distributions to the true posterior distribution we are interested in. To this end, the next section introduces scoring rules and expected scores. 

Before proceeding, we define some general notation. We will continue to use a hat to denote objects related to approximate posteriors and use $\sy$ to denote simulated data. The expectation of $f$ with respect to probability distribution $Q$ is written as~$\E_{\theta \sim Q}[f(\theta)]$ or $Q(f)$. The degenerate probability measure at $x$ is denoted by $\delta_x$. If $Q_1$ and $Q_2$ are measures where~$Q_1$ is dominated by $Q_2$ we write~$Q_1 \ll Q_2$. The Euclidean norm is denoted by $\Vert \cdot \Vert_2$.

\subsection{Bayesian Score Calibration}

Let $S:\probfam \times \Theta \rightarrow \Real \cup \{-\infty,\infty\}$ be a scoring rule for the class of probability distributions~$\probfam$ where~$U,V \in \probfam$ are defined on the measurable space~$(\Theta,\vartheta)$. A scoring rule compares a (single) observation $\theta$ to the probabilistic prediction $U$ by evaluating $S(U,\theta)$. The expected score under $V$ is defined as
$S(U,V) = \E_{\theta \sim V}\left[ S(U,\theta) \right]$. A scoring rule $S$ is strictly proper in~$\probfam$ if $S(V,V) \geq S(U,V)$ for all $U,V \in \probfam$ and equality holds if and only if $U = V$. We refer to \citet{gneiting2007strictly} for a review of scoring rules in statistics. We use the expected score $S(U,V)$ to define a discrepancy between an adjusted approximate posterior~$U$ and the true posterior $V(\cdot) = \Pi(\cdotmid y)$. 

To motivate our use of scoring rules, consider the variational problem
\begin{equation}\label{eq:posterioroptim}
    \max_{U \in \probfam} \E_{\theta \sim \Pi(\cdot\mid y)}\left[ S(U,\theta) \right],
\end{equation}
with optimal distribution $U^\star$. If $S$ is strictly proper and the class of distributions $\probfam$ is rich enough, \ie\ $\Pi(\cdotmid y) \in \probfam$ for fixed data $y$, then we recover the posterior as the optimal distribution uniquely, that is $U^{\star}(\cdot) = \Pi(\cdotmid y)$. Unfortunately, the expectation in \eqref{eq:posterioroptim} is intractable due to its expression with $\Pi(\cdotmid y)$. To circumvent this, we instead consider averaging the objective function over some distribution~$Q$ on~$(\mathsf{Y}, \mathcal{Y})$, leading to the new optimization problem
\begin{equation}\label{eq:avgposterioroptim}
    \max_{K \in \kernfam} \E_{\sy\sim Q} \E_{\theta \sim \Pi(\cdot\mid \sy)}\left[ S(K(\cdotmid \sy),\theta) \right],
\end{equation}
where $K(\cdotmid \sy)$ is now a kernel defined for $\sy \in \mathsf{Y}$, and $\kernfam$ is a family of Markov kernels. If~$\kernfam$ is sufficiently rich the optimal kernel at $y$, $K^{\star}(\cdotmid y)$, will be the posterior $\Pi(\cdotmid y)$.
\begin{definition}[Sufficiently rich kernel family]
Let $\kernfam$ be a family of Markov kernels, $\probfam$ be a class of probability measures, $Q$ be a probability measure on $(\mathsf{Y},\mathcal{Y})$, and $\Pi(\cdotmid \sy)$ be the true posterior at $\sy$. We say $\kernfam$ is sufficiently rich with respect to $(Q,\probfam)$ if for all $K \in \kernfam$, $K(\cdotmid \sy) \in \probfam$ almost surely and there exists $K \in \kernfam$ such that $K(\cdotmid \sy) = \Pi(\cdotmid \sy)$ almost surely, where $Q$ is the law of $\tilde{y}$.
\end{definition}

The maximization problem in \eqref{eq:avgposterioroptim}
is significantly more difficult than that of \eqref{eq:posterioroptim} as it involves learning the form of a Markov kernel dependent on any data generated by~$Q$, rather than a probability distribution (\ie\ with data set fixed). However, unlike \eqref{eq:posterioroptim}, the new problem \eqref{eq:avgposterioroptim} can be translated into a tractable optimization as described in Theorem~\ref{th:SBIresult}.

\begin{theorem}\label{th:SBIresult}
Consider a strictly proper scoring rule $S$ relative to the class of distributions~$\probfam$, and \textit{importance distribution} $\bar\Pi$ with Radon–Nikodym derivative $r = \dd \Pi / \dd \bar \Pi$  where $\Pi$ is the prior.   Let $v:\mathsf{Y} \rightarrow [0,\infty)$ and define $Q$ by change of measure $Q(\dd \sy) = P(\dd \sy)v(\sy)/P(v)$ such that the normalising constant $P(v) \in (0, \infty)$, where $P$ is the marginal distribution of the data. Assume the true posterior $\Pi(\cdotmid \sy) \in \probfam$ almost surely, where $Q$ is the law of $\tilde{y}$. Let the optimal Markov kernel~$K^{\star}$ be
\begin{equation}\label{eq:tracoptim}
        K^{\star} \equiv \argmax_{K \in \kernfam} \E_{\theta \sim \bar\Pi} \E_{\sy \sim P(\cdot\mid \theta)}\left[w(\theta, \sy) S(K(\cdotmid \sy),\theta) \right], \quad w(\theta, \sy) = r(\theta) v(\sy).
\end{equation}
If the family of kernels~$\kernfam$ is sufficiently rich with respect to $(Q,\probfam)$ then $K^{\star}(\cdotmid \sy) = \Pi(\cdotmid \sy)$ almost surely.
\end{theorem}
A proof for Theorem~\ref{th:SBIresult} is provided in Appendix~\ref{pr:SBIresult}.
\begin{remark}
We can interpret the optimal $K^{\star}$ in Theorem~\ref{th:SBIresult} as recovering the true posterior for $\sy \sim Q$, that is $K^{\star}(\cdotmid \sy) = \Pi(\cdotmid \sy)$ for any $\sy$ in the support of $P$ (smallest set with probability one) such that $v(\sy) > 0$. 
\end{remark}
\begin{remark}
A special case of Theorem~\ref{th:SBIresult} is stated by \citet{pacchiardi2022likelihood} where $Q = P$, the marginal distribution of the data. Considering $Q \neq P$ is crucial for establishing our subsequent results. Further, \citet{lueckmann2017flex} consider the case where the scoring rule is the log-probability score and $v(\sy) = k(\sy,y)$, a kernel measuring the discrepancy between the simulated and observed data. 
\end{remark}
\begin{remark}
The objective function in \eqref{eq:tracoptim} can also be seen as an amortized variational optimization problem. However, we take the perspective of targeting a fixed data set. We expect the function $v$ to be very useful in this setting, but do not explore this further here.
\end{remark}
Theorem~\ref{th:SBIresult} changes the order of expectation by noting the joint distribution of $(\theta, \sy)$ can be represented by the marginal distribution of $\sy$ and conditional distribution of $\theta$ given $\sy$ or vice versa. It also uses an importance distribution,~$\bar\Pi$, instead of the prior, $\Pi$. As simulators for the importance distribution (or prior) and data-generating process are assumed to be available, it is possible to estimate the objective function of \eqref{eq:tracoptim} using Monte Carlo approximations.

The weighting function $w$ is an importance sampling correction but also includes an additional component, $v$. The function $v$ describes a change of measure for the marginal distribution $P$ and represents the flexibility in $Q$ such that the optimization problems in \eqref{eq:avgposterioroptim} and \eqref{eq:tracoptim} remain equivalent. Overall then, Theorem~\ref{th:SBIresult} tells us we have the freedom to choose the importance distribution $\bar\Pi$ and change of measure $v$, in principle, without affecting the optimization. 

In practice it may be difficult to verify the conditions of Theorem~\ref{th:SBIresult} that depend on the choice of scoring rule, kernel family, importance distribution, and function $v$. We discuss these concerns, our choices, and practical consequences in the context of intractable posteriors in Sections~\ref{sec:energyscore}--\ref{sec:practicalweights}. We also provide a diagnostic tool to monitor violations of these conditions and failures in the Monte Carlo approximation of the objective function in Section~\ref{sec:calcheck}.

\subsection{Bayesian Score Calibration Algorithm}\label{sec:pseudocode}

The pseudo-code for Bayesian score calibration is detailed in Algorithm~\ref{alg:method}, where we use a Monte Carlo approximation of the optimization objective in \eqref{eq:transfoptim} and the subsequent sections describe implementation details and justifications. We assume that the vector of parameters~$\theta$ has support on $\Real^{d}$. Should only a subset of parameters in $\theta$ need correcting, Steps 5 and 6 can proceed using only this subset. Steps 5 and 6 can also be performed element-wise if correcting the joint distribution is unnecessary. When $\theta \in \Theta \subset \Real^{d}$ (a strict subset) we use an invertible transformation to map $\theta$ to $\Real^{d}$ in our examples in Section~\ref{sec:examples}. In this case, some care needs to be taken to ensure the weights are calculated correctly in Step~3. After performing the adjustment we can transform back to the original space.

\begin{algorithm}
	\caption{{Bayesian score calibration using approximate models \label{alg:method} } 
		\vspace{2mm}
		\newline
		{\em Inputs:} 
		Number of calibration data sets $M$, number of Monte Carlo samples $N$, importance distribution $\bar\Pi$, approximate posterior model $\hat\Pi$, scoring rule $S$, transformation function family $\mathcal{F}$, observed data set $y$, clipping level $\alpha\in [0,1]$. Optional: stabilizing function $v$ (otherwise unit valued).
		\vspace{2mm}
		\newline
		{\em Outputs:} Estimated optimal transformation function $f^*$ and samples from the adjusted approximate posterior based on observed data set $y$.
	}\vspace{2mm}
    \begin{enumerate}
    \item For $m \in \{1,\ldots, M\}$
	\begin{enumerate}[(i)]
		\item Generate calibration samples, $\bar\theta^{(m)} \sim \bar{\Pi}$.
		\item Generate calibration data sets, $\sy^{(m)} \sim P( \cdotmid \bar\theta^{(m)})$.
        \item Calculate weights, $\displaystyle w^{(m)} = \frac{\pi(\bar\theta^{(m)})}{\bar\pi(\bar\theta^{(m)})}v(\sy^{(m)})$ if $\alpha < 1$, else $w^{(m)} = 1$ if $\alpha = 1$.
		\item Calculate (or sample from) the approximate posteriors $\hat{\Pi}(\cdotmid \sy^{(m)})$.
        \item[] If sampling, then $\hat\theta^{(m)}_i \sim \hat{\Pi}(\cdotmid \sy^{(m)})$ for $i \in \{1, \ldots, N\}$.
        \end{enumerate}
        \item Clip weights, $w^{(m)}_\text{clip} = \min\{w^{(m)},q_{1-\alpha}\}$ for $m \in \{1,\ldots,M\}$,
        where $q_{1-\alpha}$ is the ${100(1-\alpha)}$\% empirical quantile of the weights.
		\item Solve the optimization,  $f^{\star} = \argmax_{f \in \mathcal{F}}              \sum_{m=1}^{M} w^{(m)}_\text{clip} S(\push{f}\hat{\Pi}(\cdotmid \sy^{(m)}), \bar{\theta}^{(m)})$.
  \item[] If using the energy score, use $\{\hat\theta^{(m)}_i\}_{i=1}^N$ from Step~(iv) to estimate $S$ as per \eqref{eq:mcscore}.
		\item Generate approximate adjusted samples from pushforward $\push{f}^{\star}\hat\Pi(\cdotmid y)$ by calculating $f^\star(\theta)$ where $\theta \sim \hat\Pi(\cdotmid y)$ for the desired number of samples.
	\end{enumerate}
\end{algorithm}

\subsection{The Energy Score}\label{sec:energyscore}

Thus far, we have considered a generic strictly proper scoring rule, $S$, as the propriety of our method does not rely on a specific scoring rule. For the remainder of the paper we will focus on the so-called energy score \citep[Section 4.3,][]{gneiting2007strictly} defined as
\begin{equation*}\label{eq:energyscore}
    S(U,\theta) = \frac{1}{2}\E_{u,u^{\prime} \sim U}\Vert u - u^{\prime}\Vert^{\beta}_2  - \E_{u \sim U}\Vert u - \theta \Vert^{\beta}_2,
\end{equation*}
for distribution $U$ on $(\Theta, \vartheta)$, $\theta \in \Theta$, and fixed $\beta \in (0,2)$. Note that $u$ and $u^\prime$ are independent realizations from $U$.  We find that $\beta = 1$ gives good empirical performance and note that the energy score at this value is a multivariate generalization of the continuous ranked probability score \citep{gneiting2007strictly}. The energy score is a strictly proper scoring rule for the class of Borel probability measures on
$\Theta = \Real^d$ where $\E_{u \sim U}\Vert u \Vert^{\beta}_2$ is finite. Hence, using $\beta = 1$ assumes the posterior has finite mean.

The energy score is appealing as it can be approximated using Monte Carlo methods. Assuming we can  generate $N$ samples from~$U$, as in our case, then it is possible to construct an approximation to this scoring rule as 
\begin{equation}\label{eq:mcscore}
    \hat{S}(U,\theta) = \frac{1}{ N} \sum_{i=1}^{N} \left(\frac{1}{2}  \left\Vert u_i - u_{k_i} \right\Vert^\beta_2 - \left\Vert u_i - \theta \right\Vert^\beta_2 \right), 
\end{equation}
where $u_{i} \sim U$ for $i \in \{1,\ldots,N\}$ and $k$ is a random variable, uniformly distributed over permutation vectors of length $N$.

If samples from $U$ can be generated exactly, then the approximation will be unbiased and consistent if $\E_{u \sim U}\Vert u \Vert^\beta_2 < \infty$, while inexact Monte Carlo samples from $U$ (generated by SMC or MCMC for example) will lead to a consistent estimator under the same condition.

\subsection{Approximate Posterior Transformations}

Bayesian score calibration requires a family of kernels to optimize over, $\kernfam$. We consider the family defined by conditional deterministic transformations of random variables drawn from approximate posterior distributions, that is, a pushforward measure\footnote{Also known as ``transformations of measures'' \citep[see for example,][p.\ 185]{billingsley1995probability}.} of the approximate posterior conditional on the data. This choice of kernel family allows efficient inference in our current context of expensive or intractable likelihoods, when only Monte Carlo samples from the approximate posterior are available.

We use a class of pushforward kernels, analogous to pushforward measures but with additional parameters. Consider a probability measure $\nu$ on $(\Theta, \vartheta)$. We write the pushforward measure of $g$ on~$\nu$ as $\push{g}\nu$ when $g$ is a measurable function on $(\Theta, \vartheta)$. Now consider the Markov kernel $M$ from $(\mathsf{Y},\mathcal{Y})$ to $(\Theta,\vartheta)$ and function $f: \mathsf{Y} \times \Theta \rightarrow  \Theta$ such that $f_y(\cdot) = f(y, \cdot)$ is a measurable function on $(\Theta, \vartheta)$ for each $y \in \mathsf{Y}$. We define the pushforward kernel $\push{f}M$
by stating the pushforward kernel emits (i) a conditional probability measure $\push{f}M(\cdotmid y) = M(f_y^{-1}(\cdot) \cd y)$ for fixed $y \in \mathsf{Y}$ and (ii) a function $\push{f}M(B\cd \cdot\,) = M(f_y^{-1}(B) \cd \cdot\,)$ for fixed $B \in \vartheta$. The dependence of the function $f$ on $y$ allows information from the data to inform the transformation.

The family of kernels we consider can be described as $\kernfam = \{\push{f}\hat\Pi: f \in \mathcal{F}\}$ where $\mathcal{F}$ is some family of functions.  Under such a family of kernels, we now express our idealized optimization problem as
\begin{equation}\label{eq:transfoptim}
        f^\star = \argmax_{f \in \mathcal{F}} \E_{\theta \sim \bar\Pi} \E_{\sy \sim P(\cdot\mid \theta)}\left[ w(\theta,\sy) S(\push{f}\hat\Pi(\cdotmid \sy),\theta) \right]. 
\end{equation}
One appeal of this family of pushforward (approximate) posteriors is that we only need to sample approximate draws from~$\hat\Pi(\cdotmid \sy)$ once for each~$\sy$, after which samples from~$\push{f}\hat\Pi(\cdotmid \sy)$ can be generated by applying the deterministic transformation to the set of approximate draws.

If the approximate model is computationally inexpensive to fit, then generating samples with the pushforward will also be inexpensive. This cost is also predetermined, since we fix the number of calibration data sets and the approximate posteriors only need to be learned (or sampled from) once. Moreover, once the transformation $f^\star$ is found, samples from the adjusted approximate posterior can be generated by drawing from the approximate model with observed data~$y$ and applying $f^\star$.

\subsubsection{Kernel Richness from Approximate Posterior Transformations}
To recover the true posterior, Theorem~\ref{th:SBIresult} requires that the family of kernels is sufficiently rich.
A transformation family $\mathcal{F}$ will be sufficiently rich if there exists $f \in \mathcal{F}$ such that $\push{f}\hat\Pi(\cdotmid \sy) = \Pi(\cdotmid \sy)$ almost surely for~$\sy \sim Q$. Therefore, the richness will depend on the class $\mathcal{F}$ and the approximate posterior $\hat\Pi(\cdotmid \sy)$. In applications it may be difficult to specify practical families that meet this criterion. Specifically, in our context of expensive model simulators we judge a practical family as one that is parametric with relatively few parameters---thus requiring fewer calibration data sets to be simulated (and samples from each approximate posterior).

We can also consider the richness of certain approximate posterior transformation families asymptotically. Consider realizations of the target posterior $\theta_{\sy,n} \sim \Pi(\cdotmid\sy_{1:n})$ and approximate posterior $\hat{\theta}_{\sy,n} \sim \hat\Pi(\cdotmid\sy_{1:n})$ for some $\sy_{1:n} \sim P(\cdotmid \bar\theta)$, such that
\begin{equation*}
    \sqrt{n}(\theta_{\sy,n} - \mu_{\bar\theta}) \rightarrow \dNorm(0, \Sigma_{\bar\theta}),\quad
    \sqrt{n}(\hat\theta_{\sy,n} - \hat\mu_{\bar\theta}) \rightarrow \dNorm(0, \hat\Sigma_{\bar\theta}),
\end{equation*} as $n \rightarrow \infty$ in distribution for some $\mu_{\bar\theta}, \hat\mu_{\bar\theta} \in \Real^d$ and fixed $\bar\theta\in \Real^d$. Choosing the transformation $f_{\bar\theta}(\theta) = L_{\bar\theta}[\theta - \hat\mu_{\bar\theta}] + \hat\mu_{\bar\theta} + b_{\bar\theta}$ ensures that 
\begin{equation*}
    \sqrt{n}(f_{\bar\theta}(\hat\theta_{\sy,n}) - \hat\mu_{\bar\theta} - b_{\bar\theta}) \rightarrow \dNorm(0, L_{\bar\theta}\hat\Sigma_{\bar\theta}L_{\bar\theta}^\top),
\end{equation*}
for some $b_{\bar\theta} \in \Real^d$ and $L_{\bar\theta} \in \Real^{d\times d}$. To recover the true posterior asymptotically with our method (by ensuring sufficient richness) we require $b^\star_{\bar\theta} = \mu_{\bar\theta} - \hat\mu_{\bar\theta}$ and $L^\star_{\bar\theta} = \Sigma^{1/2}_{\bar\theta}\hat\Sigma^{-1/2}_{\bar\theta}$ with $\bar\theta$ varying. As such, $\mathcal{F}$ must contain the function~$f_{\bar\theta}(\theta) = L^\star_{\bar\theta}(\theta - \hat\mu_{\bar\theta}) + \hat\mu_{\bar\theta} + b^\star_{\bar\theta}$ where $\hat\mu_{\bar\theta}$, $b^\star_{\bar\theta}$, and $L^\star_{\bar\theta}$ vary in $\bar\theta$. Practically speaking, $\hat\mu_{\bar\theta}$ can be estimated from the approximate posterior and~$\bar\theta$ can be estimated from the simulated data $\sy$. Hence both are conditional on $\sy$, yielding the approximate posterior mean $\hat\mu_{\sy}$ and estimator $\theta^\ast_{\sy}$ respectively. In this case, the class 
\begin{equation}\label{eq:asysuffnormal}
\mathcal{F} = \{f: f_{\sy}(\theta) = L(\theta^\ast_{\sy})[\theta - \hat\mu_{\sy}] + \hat\mu_{\sy} + b(\theta^\ast_{\sy}), b \in \mathcal{B}, L \in \mathcal{L}\},    
\end{equation}
would define an (asymptotically) sufficiently rich family of transformations if $\hat\mu_{\sy_{1:n}} \rightarrow \hat\mu_{\bar\theta}$ and $\theta^\ast_{\sy_{1:n}} \rightarrow \bar\theta$ as $n \rightarrow \infty$, and $\bar\theta \mapsto b^\star_{\bar\theta} \in \mathcal{B}$ and $\bar\theta \mapsto L^\star_{\bar\theta} \in \mathcal{L}$. Thus simplifying an (asymptotically) sufficiently rich class to affine functions in $\theta$, where $L$ and $b$ are only functions of a consistent estimator of $\bar\theta$.

\subsubsection{Relative Moment-Correcting Transformation}\label{sec:transform}

For this paper, we choose to use a simple transformation that corrects the location and scale of the approximate posterior under each simulated data set by the same relative amount. Assuming each approximate posterior only needs a fixed location-scale correction is a strong assumption, slightly stronger than the form described in the previous section, but empirically we find this to be a pragmatic and effective choice. Moreover, we show how to monitor and detect violations of this assumption in Section~\ref{sec:calcheck} and justify its use asymptotically.  
Using more flexible transformation families would allow for the correction of poorer approximate posterior distributions, but also require more draws from the (potentially expensive) data-generating process to estimate the transformation. We leave the exploration of more flexible transformation families for future work.

We are motivated to consider moment-matching transformations \citep[see][for example]{Warne2022,lei2011moment,Sun2016momentmatching} due to the class of asymptotically sufficiently richness of kernels derived in \eqref{eq:asysuffnormal}.  However, instead of matching moments between two random variables, we correct multiple random variables by the same relative amounts. Let the mean and covariance of a particular approximate posterior, $\hat\Pi(\cdotmid \sy)$, be $\hat{\mu}(\sy)$ and $\hat{\Sigma}(\sy)$ for some data set $\sy$ and $\Theta \subseteq \Real^d$. We denote the relative change in location and covariance by $b \in \Real^d$ and $A \in \Real^{d\times d}$ respectively, where $A$ is the decomposition of a positive definite matrix $B$, that is $AA^\top = B$. The transformation is applied to realizations from the approximate posterior, $\theta \sim \hat\Pi(\cdotmid \sy)$, as
\begin{equation}\label{eq:momcorr}
    f(\sy, \theta) = A[\theta - \hat{\mu}(\sy)] + \hat{\mu}(\sy)+ b,
\end{equation}
for $\theta \in \Theta$. The mean and covariance of the adjusted approximate posterior, $\push{f}\hat\Pi(\cdotmid \sy)$, is~$\hat\mu_f(\sy) = \hat{\mu}(\sy) + b$ and $\hat\Sigma_f(\sy) = A \hat{\Sigma}(\sy) A^{\top}$, respectively. We consider the Eigen decomposition of $B$ to parameterize $A$. In particular, we take $A = VD^{1/2}$ such that $VV^\top = I$ and $D$ is a strictly positive diagonal matrix.

Unlike the family in \eqref{eq:asysuffnormal}, the relative moment-correcting transformations we consider restrict~$b$ and $A$ to be equal for all data sets $\sy$. To justify this, we make the following observation. If we use a function $v$ such that $v(\sy_{1:n}) \rightarrow \delta_{y_{1:n}}(\sy_{1:n})$ as $n\rightarrow \infty$, then our optimization \eqref{eq:tracoptim} will simplify to the original problem \eqref{eq:posterioroptim} asymptotically. In this regime,~$b$ and~$A$ will transform the approximate posterior of a single data set, $y_{1:n}$, correcting its mean and variance. Therefore, we can use $v$ to focus our calibration on the data at hand with the aim of making the optimization less sensitive to the insufficiency of $\kernfam$, and asymptotically sufficiently rich. Our default choice of weights, which we discuss in Section~\ref{sec:practicalweights}, imply such a property for $v$.

\subsection{Choice of Weighting Function}\label{sec:practicalweights}

There are two components of the weighting function $w$ that can be chosen to refine our estimate of the optimal approximate posterior through \eqref{eq:tracoptim}. The first is the importance distribution $\bar\Pi$ which we use to concentrate the samples of~$\theta$ around likely values of the posterior distribution conditional on the observed data $y$. The second choice is the change of measure function $v$, or \textit{stabilizing function}, which we use to stabilize the weighting function~$w$ after choosing~$\bar\Pi$. Using an importance distribution $\bar\Pi$ has been considered previously \citep[for example in][]{lueckmann2017flex,pacchiardi2022likelihood} but with the inclusion of a generic~$v$ in Theorem~\ref{th:SBIresult}, the importance weights can be stabilized in various ways. We discuss idealized and practical weight functions further in Appendix~\ref{sec:weightfunctions}.

We explore weight truncation, or clipping \citep{ionides2008truncated}, to ensure finite variance of the weights. In general, clipping the empirical weights is achieved by 
\begin{equation}\label{eq:clip}
    w^{(m)}_{\text{clip}} = \min\{w^{(m)},q_{1-\alpha}\},\quad m \in \{1,\ldots, M\},
\end{equation}
where the truncation value $q_{1-\alpha}$ is the $100(1-\alpha)$\% empirical quantile based on weights $\{w^{(m)}\}_{m=1}^M$. Letting $\alpha \in [0,1]$ depend on~$M$ such that $\alpha \rightarrow 0$ as $M \rightarrow \infty$ is sufficient for asymptotic consistency. However, full clipping \ie\ $\alpha = 1$ or equivalently \textit{unit weights}, will not satisfy this. Instead, we establish asymptotic consistency (in the size of the data set) for unit weights next, and demonstrate good empirical performance in Section~\ref{sec:examples}.

\subsubsection{Unit Weighting Function} \label{sec:uweights}

In this section we will consider the effect of approximating the weight function $w(\theta, \sy)$ with unit weights, \ie\ $\hat{w} = 1$. This can also be viewed as clipping with $\alpha = 1$ in \eqref{eq:clip}.  We consider weights of the form $w(\theta,\sy_{1:n}) = v(\sy_{1:n}) \pi(\theta) / \bar\pi(\theta)$
where the number of observations $n\rightarrow\infty$ and $r(\theta) = \pi(\theta) / \bar\pi(\theta)$. The results of this section are possible due to the stability function, $v$, which is free to be chosen without affecting the validity of the method, as established by Theorem~\ref{th:SBIresult}. We first consider the consistency of the unit weights, when a consistent estimator $\theta^{\ast}$ exists.
\begin{theorem}\label{th:unitweights-as}
Let $g(x) = \bar \pi(x) / \pi(x)$ for $x \in \Theta$. If there exists an estimator $\tilde{\theta}^{\ast}_n \equiv \theta^{\ast}(\sy_{1:n})$ such that  $\tilde{\theta}^{\ast}_n \asarrow z$ as $n \rightarrow \infty$ when~$\tilde y_i \simiid P(\cdotmid z)$ for $z \in \Theta$, and $g$ is positive and continuous at $z$ then the error when using $\hat{w} = 1$ satisfies
\begin{equation*}
    \hat{w} - w(\theta,\sy_{1:n}) \asarrow 0,
\end{equation*}
as $n \rightarrow \infty$ with choice of stabilizing function $v(\sy_{1:n}) = g(\tilde{\theta}^\ast_n)$.
\end{theorem}
Theorem~\ref{th:unitweights-as} establishes a strong consistency result; unit weights are a large sample approximation to the theoretically correct weights for some choice of stabilizing function $v$. A corresponding weak consistency result follows analogously. Crucially, a consistent estimator is not required to implement our method in practice. The existence of such an estimator simply ensures that unit weights are a valid choice asymptotically. From a practical standpoint, unit weights effectively remove the importance sampling component of \eqref{eq:tracoptim}, ensuring our approach is numerically stable, and practically appealing. A proof of Theorem~\ref{th:unitweights-as} is provided in Appendix~\ref{pr:unitweights-as}, whilst a central limit theorem for the unit weight approximation appears in Appendix~\ref{thpr:unitweights-clt}. 

\subsubsection{Importance Distribution}\label{sec:importance}
The importance distribution $\bar\Pi$ can be chosen to focus the calibration on regions of $\Theta$. For this paper we use the approximate posterior $\hat\Pi(\cdotmid y)$ to choose $\bar\Pi$. Specifically, we use the scale transformation
\begin{equation}\label{eq:impdist}
    \bar\Pi(\dd \theta) \propto \hat\pi(D^{-1}(\theta - \hat\mu) + \hat\mu \cd y)\dd \theta,
\end{equation}
where $\hat\pi(\cdotmid y)$ is the density of the approximate posterior $\hat\Pi(\cdotmid y)$, $D$ is a positive-definite diagonal matrix used to inflate the variance of the approximate posterior, and $\hat\mu$ is the estimated mean of the approximate posterior. We can draw samples from $\bar\Pi$ using the transformation $D(\theta^\prime - \hat\mu) + \hat\mu$ when $\theta^\prime \sim \hat\Pi(\cdotmid y)$.

With unit weights and importance distribution \eqref{eq:impdist}, the stabilizing function from Theorem~\ref{th:unitweights-as} will have the form
\begin{align*}
    v(\sy_{1:n}) &= \hat p(y_{1:n} \cd D^{-1}(\tilde{\theta}^\ast_n - \hat\mu) + \hat\mu) c(\tilde{\theta}^\ast_n), \\
    \text{where}~ & c(\tilde{\theta}^\ast_n) \propto \frac{\pi(D^{-1}(\tilde{\theta}^\ast_n - \hat\mu) + \hat\mu)}{\pi(\tilde{\theta}^\ast_n)},
\end{align*} 
which will behave as $v(\sy_{1:n}) \rightarrow \delta_{\hat{\theta}_0}(\tilde{\theta}^\ast_n)$ for $n\rightarrow\infty$ where $\hat{\theta}_0$ is the maximum likelihood estimator (MLE) from $\hat{p}(y_{1:n} \cd \cdot\,)$ as $n\rightarrow\infty$. Asymptotically, this would reduce the support of the importance distribution $Q$ to the manifold $\mathsf{M}_n = \{\sy_{1:n} \in \mathsf{Y}^n: \tilde{\theta}^\ast_n = \hat{\theta}_0\}$ and will depend on the sufficient statistics for the true posterior. On the manifold $\mathsf{M}_n$, each $\sy_{1:n}$ will have the same limiting distribution (if it exists) as the consistent estimator $\tilde{\theta}^\ast_n$ for each $\sy_{1:n}$ is constrained to be equal. This justifies our use of constant~$b$ and~$A$ to define our approximate posterior transformation family (when using unit weights) as there is only one true and one approximate posterior to correct for in this regime. Hence, transformations defined by \eqref{eq:momcorr} are asymptotically sufficiently rich under some mild conditions when using unit weights.

\subsection{Calibration Diagnostic}\label{sec:calcheck}

To assist using Bayesian score calibration in practice, we suggest a performance diagnostic to warn users if the adjusted approximate posterior is unsuitable for inference. The diagnostic detects when the learned transformation does not adequately correct the approximate posteriors from the calibration data sets. The diagnostic can be computed with trivial expense as it requires no additional simulations.

To elaborate, whilst executing Algorithm~\ref{alg:method}, we have access to the true data-generating parameter value $\bar\theta^{(m)}$ for each adjusted approximate posterior, $\push{f}^{\star}\hat\Pi(\cdotmid \sy^{(m)})$. Therefore, we can measure how these adjusted approximate posteriors perform (on average) relative to this true value. Various metrics could be used for this task, but we find the empirical coverage probabilities for varying nominal levels of coverage to be suitable.

Specifically, if $\text{Cr}(U,\rho)$ is a $(100\times\rho)$\% credible interval (or highest probability region) for distribution $U$ then we calculate the achieved coverage (AC) by estimating
\begin{equation}\label{eq:diagnostic}
    \text{AC}(\rho) = \Pr\left[\bar\theta \in \text{Cr}(\push{f}^{\star}\hat\Pi(\cdotmid \sy),\rho)\right],\quad\text{where}~\sy \sim P( \cdotmid \bar\theta).
\end{equation} 
for a sequence of $\rho \in (0,1)$, using pairs of $(\bar\theta^{(m)}, \{f^{\star}(\hat\theta^{(m)}_i)\}_{i=1}^{N})$ for $m \in \{1,\ldots,M\}$ generated by Algorithm~\ref{alg:method}. We refer to $\rho$ as the target coverage. Whilst this diagnostic can be calculated on the joint distribution of the posterior, for simplicity we will use the marginal version of \eqref{eq:diagnostic} resulting in a diagnostic for each parameter in the posterior. We forgo a multivariate diagnostic as the marginal version requires less user input (only the type of credible intervals to use) and multivariate versions are far more difficult to compute. For our experiments we use credible intervals with end-points determined by symmetric tail-probabilities and plot the miscoverage for an interval of target coverage levels $\rho \in [0.1,0.95]$. The miscoverage is $\text{MC}(\rho) = \text{AC}(\rho) - \rho$, where positive values indicate over-coverage, while negative values indicate under-coverage. 

The miscoverage diagnostic will be sensitive to failures of the method in the high probability regions of the importance distribution,~$\bar\Pi$. If one wishes to test areas outside this region or in specific areas, new pairs of transformed approximate posteriors and their data-generating value could be produced at the cost of additional computation. With respect to the importance distribution, this diagnostic will help to detect if the quality of the approximate distribution is insufficient, if the transformation family is too limited, if the weights have too high variance or if the optimization procedure otherwise fails (for example due to an insufficient number of calibration data sets). 

\subsection{Related Research}\label{sec:related}

\citet{lee2019calibration} and \citet{pmlr-v97-xing19a} develop similar calibration procedures to ours but to estimate the true coverage of approximate credible sets as a diagnostic tool or means to adjust their posterior. We adjust the approximate posterior samples directly based on their distribution rather than just correcting coverage. Relatedly,  \citet{menendez2014simultaneous} correct confidence intervals from 
approximate inference for bias and nominal coverage in the frequentist sense, and \citet{rodrigues2018recalibration} calibrate the entire approximate posterior based on similar arguments.

\citet{xing2020distortion} develop a method to transform the marginal distributions of an approximate posterior without expensive likelihood evaluation.  They estimate a distortion map, which, theoretically, transports the approximate posterior to the exact posterior (marginally).  Since the true distortion map is unavailable, \citet{xing2020distortion} learn the distortion map using simulated data sets, their associated approximate posteriors, and the true value of the parameter used to generate the data set (similar to our approach).  They fit a beta regression model to the training data, which consists of approximate CDF values as the response and the data sets (or summary statistics thereof) as the features.  \citet{xing2020distortion} learn the parameters of the beta distribution using neural networks.  Their approach ensures that the approximate posterior transformed with the estimated map reduces the Kullback--Leibler divergence to the true posterior.  However, in their examples, \citet{xing2020distortion} use $\mathcal{O}(10^6)$ simulations from the model of interest, in order to have a sufficiently large sample to train the neural network.   Another reason for the large number of model simulations is they only retain a small proportion of simulated data sets from the prior predictive distribution that are closest to the observed data, in an effort to obtain a more accurate neural network localized around the observed data.  Our method only requires generating $\mathcal{O}(10^2)$ data sets from the target model, and thus may be more suited to models where it is moderately or highly computationally costly to simulate.  Further, in their examples,  \citet{xing2020distortion} require fitting $\mathcal{O}(10^4)$ to $\mathcal{O}(10^5)$ approximate posteriors, whereas we only require $\mathcal{O}(10^2)$.  Thus our approach has a substantially reduced computational cost.

\citet{rodrigues2018recalibration} develop a calibration method based on the coverage property that was previously used in \citet{prangle2014diagnostic} as a diagnostic tool for approximate Bayesian computation \citep[ABC,][]{beaumont2002approximate,Sisson2018}.  Even though a key focus of \citet{rodrigues2018recalibration} is to adjust ABC approximations, the method can be used to recalibrate inferences from an approximate model.  Like \citet{xing2020distortion}, \citet{rodrigues2018recalibration} require a much larger number of model simulations and approximate posterior calculations compared to our approach.  Furthermore, \citet{rodrigues2018recalibration} require evaluating the CDF of posterior approximations at parameter values used to simulate from the target model.  Thus, if the surrogate model is not sufficiently accurate, the CDF may be numerically 0 or 1, and the corresponding recalibrated sample will not be finite.  The method of \citet{xing2020distortion} may also suffer from similar numerical issues.  Our approach using the energy score is numerically stable.

\citet{vandeskog2022adjusting} develop a post-processing method for posterior samples to correct composite and otherwise misspecified likelihoods that have been used for computational convenience. Their method uses a linear transformation to correct the asymptotic variance of the model at the estimated mode (or suitable point estimate). 
They show that their method can greatly improve the low coverage resulting from the initial misspecification. Their adjustment requires an analytical form for the true likelihood with first and second order derivatives. Our method does not require an analytical form for the true likelihood, nor calculable derivatives. Moreover, we derive an adjustment in the finite-sample regime. 

A related area is delayed acceptance MCMC \citep[\eg][]{sherlock2017adaptive} or SMC \citep[\eg][]{bon2021accelerating}.  In delayed acceptance methods, a proposal parameter is first screened through a Metropolis-Hastings (MH) step that depends only on the likelihood for the surrogate model.  If the proposal passes this step, it progresses to the next MH stage that depends on the likelihood of the expensive model, otherwise the proposal can be rejected quickly without probing the expensive likelihood.  Although exact Bayesian inference can be generated with delayed acceptance methods, they require a substantial number of expensive likelihood computations, which limits the speed-ups that can be achieved.  Our approach, although approximate, does not require any expensive likelihood calculations, and thus is more suited to complex models with highly computational expensive or completely intractable likelihoods. The idea of delayed acceptance is generalized with multifidelity methods in which a continuation probability function is optimized based on the receiver operating characteristic curve with the approximate model treated as a classifier for the expensive model~\citep{Prescott2020,Prescott2021,warne2022multifidelity}. 

 In other related work, \citet{Warne2022} consider two approaches, preconditioning and moment-matching methods, that exploit approximate models in an SMC setting for ABC. The preconditioning approach applies a two-stage mutation and importance resampling step that uses an approximate model to construct a more efficient proposal distribution that reduces the number of expensive stochastic simulations required. The moment-matching approach transforms particles from an approximate SMC sampler to increase particle numbers and statistical efficiency of an SMC sampler using the expensive model. Of these two methods, the moment-matching SMC approach is demonstrated to be particularly effective in practice. As a result, we use the moment-matching transformation to inform the moment-correcting transformation used in this work.

As for frequentist-based uncertainty quantification, \citet{warne2023generalised} develop a method for valid frequentist coverage for intractable likelihoods with generalized likelihood profiles, whilst \citet{muller2013risk} obtain valid frequentist properties in misspecified models based on sandwich covariance matrix adjustments.  \citet{frazier2022bayesian} use a similar adjustment to correct for a misspecified Bayesian synthetic likelihood approach to likelihood-free inference. A deliberate  misspecification of the covariance matrix in the synthetic likelihood is used to speed-up computation. Then a post-processing step compensates for the misspecification in the approximate posterior.  Their approach does not have the goal of approximating the posterior distribution for the correctly specified model, and the post-processing step performs only a covariance adjustment without any adjustment of the mean. 

Related work has considered learning the conditional density of the posterior as a neural network \citep{papamakarios2016fast,lueckmann2017flex, greenberg2019auto}. In these methods, the conditional density estimate is updated sequentially using samples from the current approximate posterior. Sequential neural posterior estimation, as it is called, is built upon by \citet{papamakarios2019seq} with a focus on learning a neural approximation to the likelihood, rather than the full posterior. We expect the theoretical framework we propose to be useful in these contexts for developing methods to stabilize the importance weights, or understanding existing attempts at this \citep[\eg][]{deistler2022truncated}.

\citet{pacchiardi2022likelihood} explore similar concepts and relate these to generative adversarial networks. Our work has related theoretical foundations, in that we use expectations over the marginal probability of the data to circumvent the intractability of the posterior, though we have a more general formulation. Moreover, we also focus on the case where our learned posterior is a correction of an approximate posterior we have access to samples from. As in \citet{pacchiardi2022likelihood} we are also concerned with scoring rules, in particular, the energy score. We refer to an extended review of machine learning approaches for likelihood-free inference surmised by \citet{cranmer2020frontier} for further reading.

\section{Examples} \label{sec:examples}

This section contains a number of empirical examples demonstrating Bayesian score calibration. A worked example to elucidate the elements of our framework is provided in Appendix~\ref{app:logit} for a Bayesian logistic regression. We defer a further three empirical examples to the supplementary materials \ref{sec:ex-gauss}--\ref{sec:lotka-abc} for brevity.

In the following examples we use $\beta = 1$ for the tuning parameter of the energy score, $M = 100$ or $M = 200$ as the number of calibration data sets and use the approximate posterior with scale inflated by a factor of $2$ as the importance distribution, unless otherwise stated. The examples in Section~\ref{sec:ex-ous} are tractable and inexpensive to run so that we can assess the performance of the calibration method against the true posterior on repeated independent data sets. The examples in Sections~\ref{sec:ex-lotka} and~\ref{sec:ex-mapk} have intractable likelihoods, so we do not perform an exact inference comparison but do compute coverage comparisons (since the truth is known) to investigate the results. A \texttt{julia} \citep{bezanson2017julia} package implementing our methods and reproducing our examples is available online, see Appendix~\ref{sec:code} for details.

For examples where we validate the performance of our model calibration procedure using new independent data sets we compare the adjusted (approximate) posterior with the original approximate posteriors, and true posteriors using the average (over the $M$ independent data sets); mean square error (MSE), bias of the posterior mean, posterior standard deviation, and the coverage rate of the nominal 90\% credible intervals.  For some marginal posterior samples $\{ \theta_j\}_{j=1}^J$, we approximate the MSE using
$\widehat{\text{MSE}} = \frac{1}{J}\sum_{j=1}^J (\theta_j - \theta)^2$,
where $\theta$ is the known scalar parameter value (\ie\ a particular component of the full parameter vector that we are adjusting).

\subsection{Ornstein–Uhlenbeck Process} \label{sec:ex-ous}

We first consider the Ornstein-Uhlenbeck (OU) process~\citep{Uhlenbeck1930}. The OU process, $\{X_t\}_{t \geq 0}$ for $X_t \in \Real$, is a mean-reverting stochastic process that is governed by the It\^{o} stochastic differential equation~(SDE)
\begin{equation}
\text{d}X_t = \gamma (\mu - X_t) \text{d}t + \sigma \text{d}W_t,
\label{eq:ousde}
\end{equation}
with mean $\mu$ and volatility of the process denoted by $\sigma$. The rate at which $X_t$ reverts to the mean is $\gamma$, whilst $W_t$ is a standard Wiener process. 

Given an initial condition, $X_0 = x_0$ at $t = 0$, we can obtain the distribution of the state at future time $T$ through the solution to the forward Kolmogorov equation (FKE) for \eqref{eq:ousde}. For the OU process, the FKE is tractable with the solution
\begin{align}
    X_T & \sim \dNorm\left(  \mu + (x_0 - \mu)e^{-\gamma T}, \frac{\sigma^2}{2\gamma}(1- e^{-2\gamma T}) \right).
    \label{eq:ousol}
\end{align}
However, we note that analytical results are not available for most SDE models and one must rely on numerical methods such as Euler-Maruyama schemes~\citep{Maruyama1955}. 

The next sections illustrate our method on two OU processes, the processes are one- and two-dimensional respectively. In the first example we approximate the likelihood using the limiting distribution of the OU process, whilst in the second we approximate the posterior distribution using variational inference. The first example uses independent transformations for each parameter and the second uses a multivariate transformation.

\subsubsection{Univariate OU Process with Limiting Distribution Approximation}
Take $X_T$ as defined in \eqref{eq:ousol} as the true model (and corresponding likelihood) in this example.  
For the observed data we take 100 independent realizations simulated from the above model with $x_0 = 10$, $\mu = 1$, $\gamma = 2$, $T=1$ and $\sigma^2 = 20$.  We assume that $x_0$ and $\gamma$ are known and we attempt to infer $\mu$ and $D = \sigma^2/2$.  We use independent priors where $\mu \sim \dNorm(0,10^2)$ and $D \sim \dExp(1/10)$ (parameterized by the rate). We sample and perform our adjustment over the space of $\log D$, but report results in the original space of $D$.  The limiting distribution $T \rightarrow \infty$ is the approximate model, $X_\infty  \sim \dNorm\left(\mu, \frac{\sigma^2}{2\gamma}\right)$, from which we define the approximate likelihood.  Clearly there will be a bias in the estimation of $\mu$.

\begin{figure}
		\centering
		\includegraphics[width=0.75\textwidth]{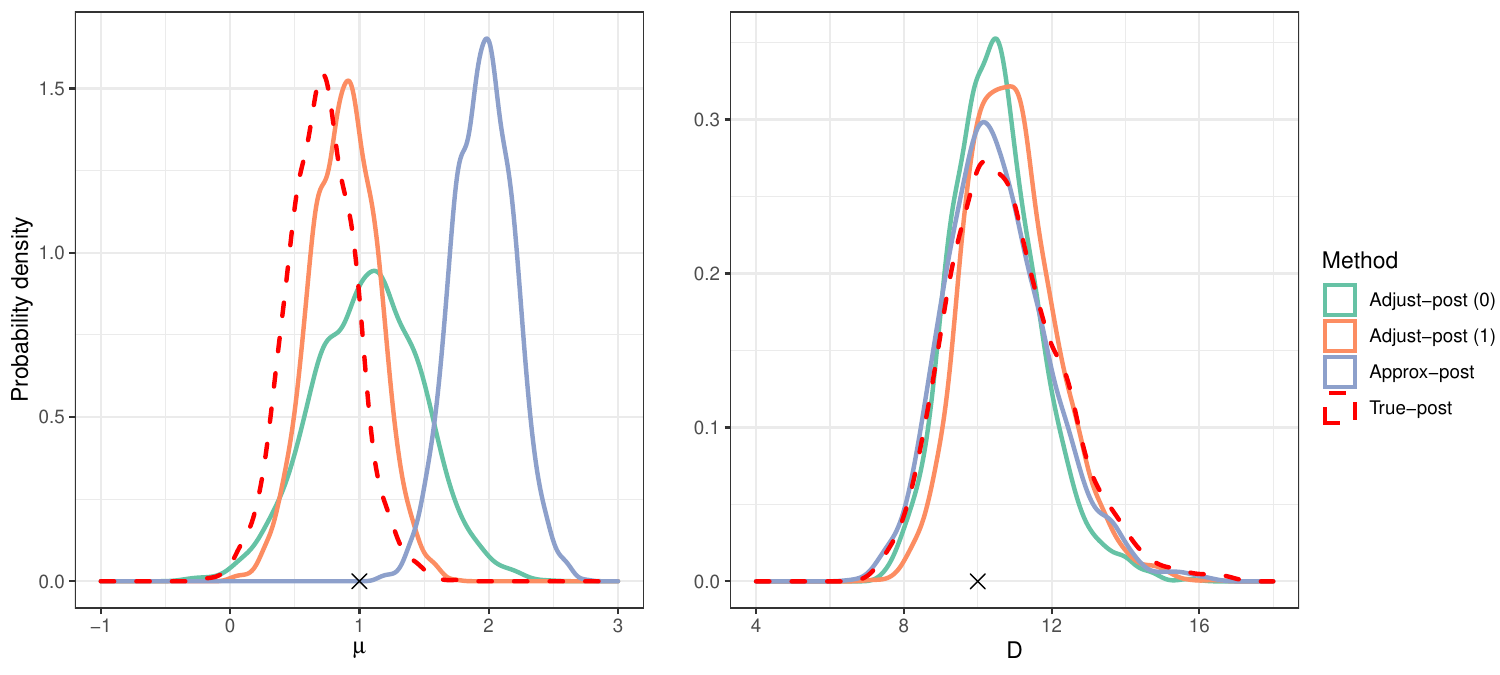}
		\caption{Univariate densities estimates of approximations to the OU Process model posterior distribution from a single simulation. The original approximate posterior (Approx-post) and adjusted posteriors (Adjust-post) with $(\alpha)$ clipping are shown with solid lines. The true posterior (True-post) is shown with a dashed line. The true generating parameter value is indicated with a cross $(\times)$.}
		\label{fig:posteriors_ou_limit}
\end{figure}

For this example we sample from the approximate and true posteriors using the \texttt{Turing.jl} library \citep{ge2018t} in \texttt{Julia} \citep{bezanson2017julia}. We use the default No-U-Turn Hamiltonian Monte Carlo algorithm \citep{hoffman2014no}. For simplicity, we set the stabilizing function $v(\sy) = 1$.  

\begin{longtable}{lrrrrrrrr}
\toprule
\multicolumn{1}{l}{} & \multicolumn{4}{c}{$\mu$} & \multicolumn{4}{c}{$D$}\\ 
\midrule
Method & MSE & Bias & SD & AC$^1$ & MSE & Bias & SD & AC$^1$\\
\midrule
Approx-post & $1.54$ & $1.21$ & $0.22$ & $0$ & $4.73$ & $0.18$ & $1.46$ & $85$ \\ 
Adjust-post (0) & $0.12$ & $0.15$ & $0.20$ & $64$ & $4.83$ & $0.28$ & $1.24$ & $72$\\ 
Adjust-post (0.5) & $0.12$ & $0.15$ & $0.23$ & $81$ & $5.08$ & $0.41$ & $1.42$ & $81$\\ 
Adjust-post (1) & $0.12$ & $0.15$ & $0.23$ & $82$ & $5.13$ & $0.42$ & $1.45$ & $83$\\ 
True-post & $0.12$ & $-0.01$ & $0.26$ & $94$ & $5.00$ & $0.37$ & $1.48$ & $85$\\ 
\bottomrule
\\
\caption{Average results for each parameter over 100 independent data sets for the univariate OU process example. The posteriors compared are the original approximate posterior (Approx-post) and adjusted posteriors (Adjust-post) with $(\alpha)$ clipping, and the true posterior (True-post). $^1$Achieved coverage with 90\% target.\label{tab:repeated_results_ou}}
\end{longtable}

Results for the mean squared error (MSE), bias, standard deviation (SD), and achieved coverage (AC, 90\%) are presented in Table \ref{tab:repeated_results_ou}.  It is clear that the approximate posterior performs poorly for $\mu$.  Despite this, the adjustment method is still able to produce results that are similar to the true posterior on average.  The approximate method already produces accurate inferences similar to the true posterior for $D$ so that the adjustment is negligible.  As an example, the posterior results based on running the adjustment process on a single data set are shown in Figure \ref{fig:posteriors_ou_limit}.

\subsubsection{Bivariate OU Process with Variational Approximation}
We can define a bivariate OU process by considering $\{Y_t\}_{t \geq 0}$ for $Y_t \in \Real^2$, such that the components are $Y_{t,1} = X_{t,1}$ and $Y_{t,2} = \rho X_{t,1} + (1-\rho)X_{t,2}$ where $X_{t,1}$ and $X_{t,2}$ are independent OU processes, conditional on shared parameters $(\mu,\gamma,\sigma)$, and governed by \eqref{eq:ousde}. The additional parameter $\rho \in [0,1]$ measures the correlation between $Y_{t,1}$ and $Y_{t,2}$.

\begin{figure}
		\centering
		\subfigure[Bivariate density contours of $\rho$ and $D$]{\includegraphics[width=0.47\textwidth]{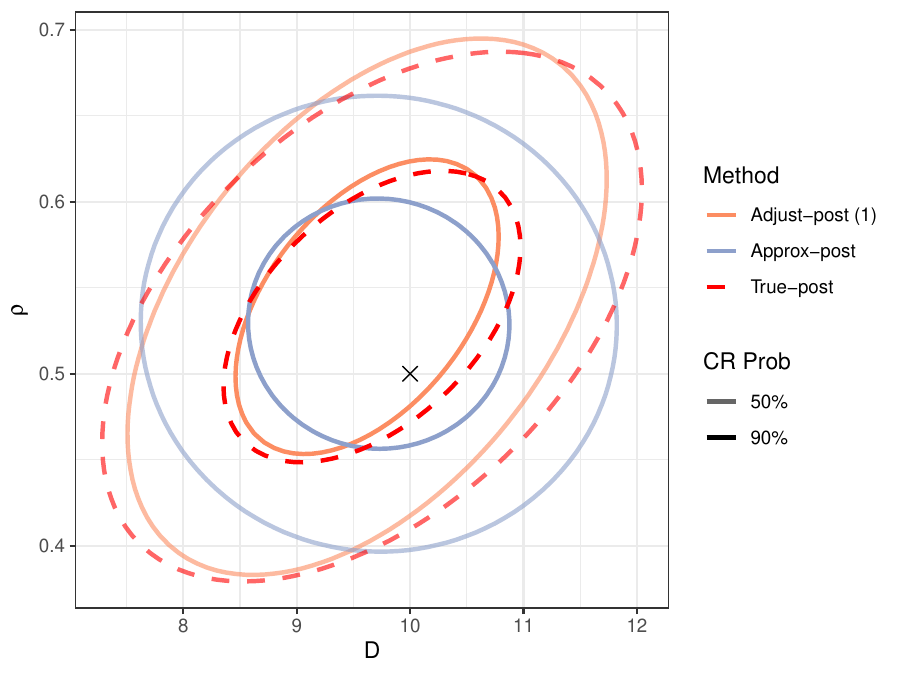}\label{figsub:results_corr_ou_contour}}
		\subfigure[Calibration diagnostic]{\includegraphics[width=0.47\textwidth]{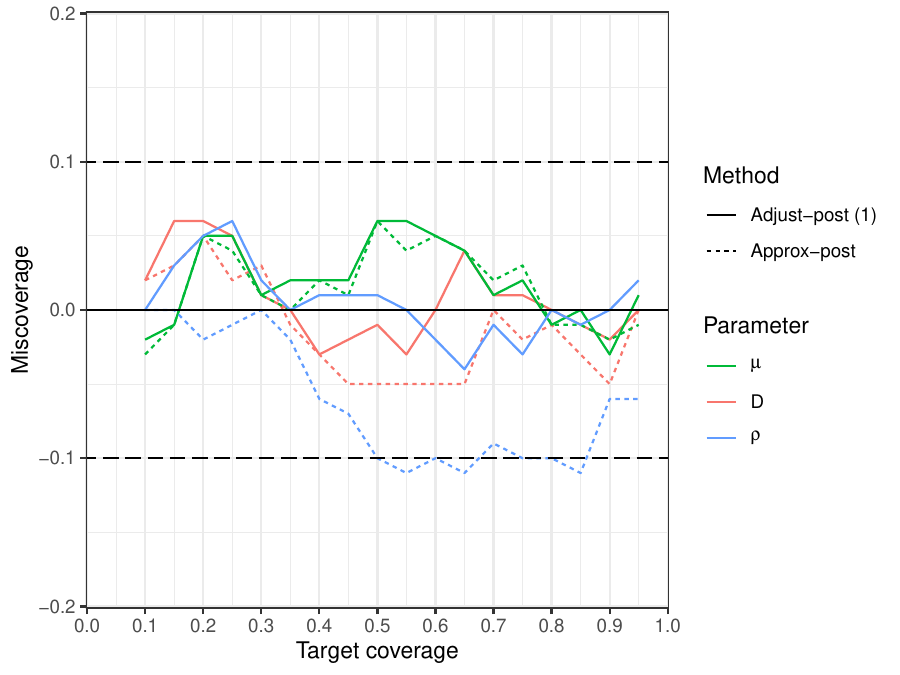}\label{figsub:results_corr_ou_calcheck}}
		\caption{Posterior summaries of the bivariate OU Process model from a single simulation. Plot (a) shows 50\% and 90\% credible region probability (CR Prob) contours from a Gaussian approximation to the bivariate density of $\rho$ and $D$. The original approximate posterior (Approx-post) and adjusted posteriors (Adjust-post) with $(\alpha)$ clipping are shown with solid lines. The true posterior (True-post) is shown with a dashed line. The true generating parameter value is indicated with a cross $(\times)$. Plot (b) is a calibration diagnostic showing the marginal miscoverage for all parameters (see Section~\ref{sec:calcheck}) for $\alpha = 1$ with $\pm 0.1$ deviation from parity shown with a dotted line. }
		\label{fig:posteriors_ou_vi}
\end{figure}

Again, we consider the true model for $X_{t,1}$ and $X_{t,2}$ to be defined by \eqref{eq:ousol}, therefore $(Y_{t,1}, Y_{t,2})$ have joint distribution that is bivariate Gaussian with correlation $\rho$. For the observed data we take 100 independent realizations simulated from the above model with $x_{0,1} = x_{0,2} = 5$, $\mu = 1$, $\gamma = 2$, $T=1$, $\sigma^2 = 20$, and $\rho = 0.5$.  We assume that $x_{0,1}$, $x_{0,2}$, and $\gamma$ are known and we attempt to infer $\mu$, $D = \sigma^2/2$, and $\rho$.  We use independent priors where $\mu \sim \dNorm(0,10^2)$, $D \sim \dExp(1/10)$, and $\rho \sim \dUnif(0,1)$. We use automatic differentiation variational inference \citep{kucukelbir2017automatic} for the approximate model with a mean-field approximation as the variational family as implemented in \texttt{Turing.jl}. We expect the correction from our method will need to introduce correlation in the posterior due to the independence inherited from the mean-field approximation.

The variational approximation of the bivariate OU posterior estimates the mean and variance well in this example. The univariate bias, MSE, and coverage metrics are very similar for the approximate, adjusted ($\alpha = 1$) and true posteriors (Table~\ref{tab:repeated_results_mv_ou} in Appendix~\ref{app:tabfigs}). However, the adjusted posterior with $\alpha = 0$ is poor due to high variance of the weights, perhaps due to the increased dimension of this example. This could be corrected using an appropriate stabilizing function, but we leave this for future research.

Despite the good univariate properties of the variational approximate posterior, the choice of a mean-field approximation cannot recover the correlation between the parameters, in particular, $\rho$ and $D$. Figure~\ref{figsub:results_corr_ou_contour} shows an example of the independence between $\rho$ and $D$ in the approximate posterior and how the adjusted posterior corrects this (from one data set). To investigate the method's ability to recover the correlation structure we monitor the empirical correlation between $\rho$ and $D$ over 100 independent trials. Adjusting the approximate posterior ($\alpha=1$) increased the mean correlation from 0.00 to 0.31, a major improvement compared to 0.41 for the true posterior. Further correlation summaries are provided in Table~\ref{tab:corr_bi_ou}.

Figure~\ref{figsub:results_corr_ou_calcheck}
illustrates the calibration diagnostic from Section~\ref{sec:calcheck} for the same simulation as Figure~\ref{figsub:results_corr_ou_contour}. We can see the learned transformation leads to a well-calibrated adjusted posterior as the miscoverage is close to the target coverage for the range considered, and improves the coverage compared to the original approximate posterior of $\rho$.

\subsection{Lotka-Volterra Model with Kalman Filter Approximation} \label{sec:ex-lotka}

We consider the Lotka-Volterra, or predator-prey, dynamics governed by a stochastic differential equation (SDE). Let $\{(X_t,Y_t)\}_{t \geq 0}$ be a continuous time stochastic process defined by the SDE
\begin{equation*}\label{eq:sde_lotka}
    \dd X_{t} = (\beta_1 X_{t} - \beta_2 X_{t} Y_{t} ) \dd t+ \sigma_1 \dd B_{t}^{1}, \quad
    \dd Y_{t} = (\beta_4 X_{t} Y_{t} - \beta_3 Y_{t} ) \dd t+ \sigma_2 \dd B_{t}^{2},
\end{equation*}
where $\{B_{t}^{k}\}_{t\geq 0}$ are independent Brownian noise processes for $k \in \{1,2\}$. For this example, we assume the pairs $(x_t,y_t)$ are observed without error at times $t \in \{0, 0.2, 0.4, \ldots, 6\}$, for a total of $n=31$ observations. We use initial values $X_0 = Y_0 = 1$ and simulate the observations with true parameter values $\beta_1 = 1.5,  \beta_2 = \beta_4 = 1.0, \beta_3 = 3.0$, and $\sigma_1 = \sigma_2 = 0.1$, using the SOSRI solver \citep{rackauckas2020stability}.

\begin{figure}
		\centering
		\subfigure[Posterior distribution comparison]{\includegraphics[width=0.47\textwidth]{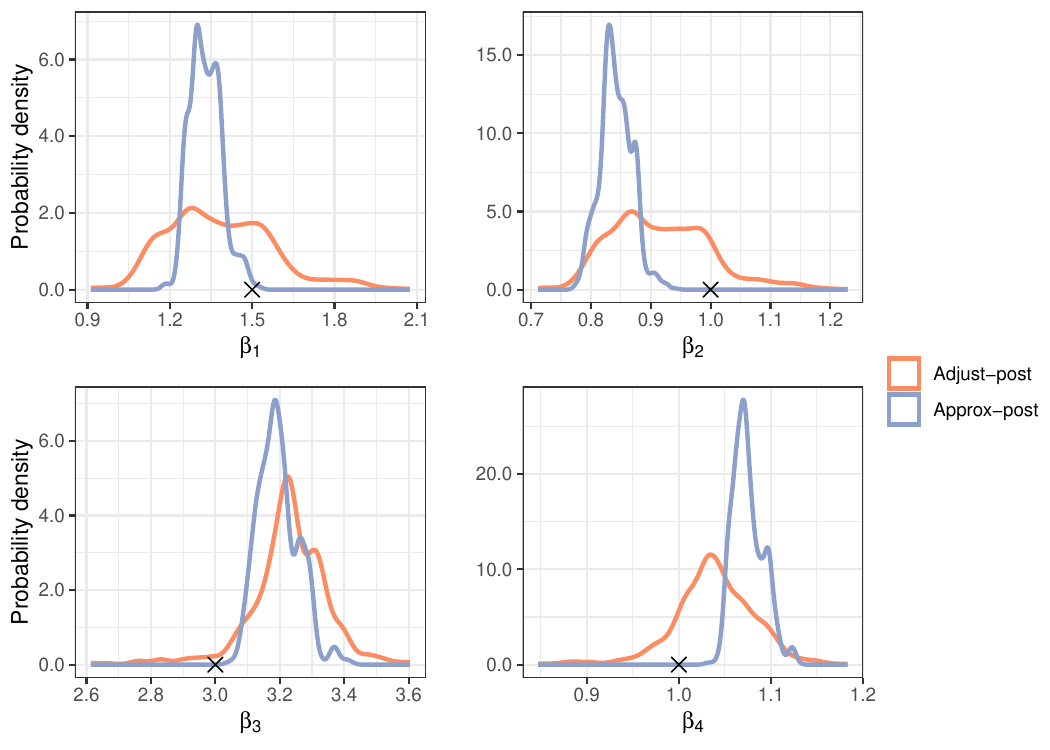}\label{figsub:results_lotka_ekf_post}}
		\subfigure[Calibration diagnostic]{\includegraphics[width=0.47\textwidth]{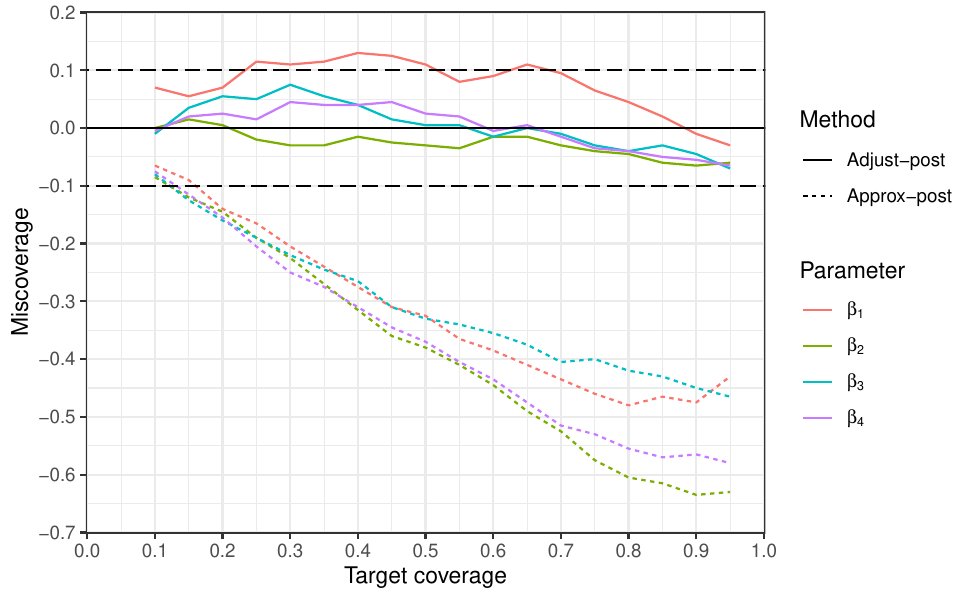}\label{figsub:figsub:results_lotka_ekf_cal}}
		\caption{Comparison of original approximate posterior (Approx-post) and adjusted posteriors (Adjust-post) for Lotka-Volterra example with EKF likelihood. Plot (a) shows the estimated marginal posterior densities, with true generating parameter value indicated with a cross $(\times)$. Plot (b) is a calibration diagnostic showing the marginal miscoverage for all parameters with $\pm 0.1$ deviation from parity shown with a dotted line.}
		\label{fig:posteriors_lotka_ekf}
\end{figure}

We define our target posterior with priors $\beta_i \sim \mathcal{U}(0.1,4)$ iid for $i\in\{1,2,3,4\}$ and $\sigma_{j} \sim \mathcal{U}(0.01,0.25)$ iid for $j\in\{1,2\}$. We use a continuous-discrete Extended Kalman Filter \citep[EKF,][]{jazwinski2007stochastic} as the approximate likelihood and draw samples using the No-U-Turn Sampler \citep[NUTS,][]{hoffman2014no} for Hamiltonian Monte Carlo \citep[HMC,][]{neal2011mcmc} algorithm from \texttt{Turing.jl}. For the prediction step, the EKF requires numerical integration of the moment equations and we use a Euler-Maruyama scheme with step-size $\Delta t = 0.05$ \citep[see][for an overview]{frogerais2011various}. The EKF also requires specification of the observation noise, which we choose to be $(x_t,y_t) \sim \mathcal{N}((X_t,Y_t), \tau^2 I)$, where the prior for the standard deviation is the positive truncated normal distribution, $\tau \sim \mathcal{N}^{+}(0,0.05^2)$, chosen to favor the low noise regime we are simulating data from.

We use $M=200$ calibration samples, unit weights (\ie\ $\alpha=1$), and draw 1000 samples for the approximate posteriors. In this example, the adjusted posterior does significantly better than the approximate posterior. Figure~\ref{figsub:figsub:results_lotka_ekf_cal} shows the marginal approximate posteriors have extreme under-coverage which is corrected by our method. In Figure~\ref{figsub:results_lotka_ekf_post} we observe strong bias in parameters~$\beta_2$ and~$\beta_4$ which is removed, whilst the variance of $\beta_1$ is inflated to ensure coverage. Whilst the marginal distribution of $\beta_3$ appears mostly unchanged, the adjustment does ensure probability mass at the true value, where the approximate posterior had none.

\subsection{Stochastic Chemical Kinetic Model} \label{sec:ex-mapk}

Finally, we demonstrate our method on a difficult biological model with intractable likelihood and parameter non-identifiability. The chemical reaction network we consider is a two-step Mitogen Activated Protein Kinase (MAPK) enzymatic cascade \citep{dhananjaneyulu2012noise}. Such reaction networks are often used as components of larger systems to model cell signaling processes \citep{brown2004statistical,oda2005comprehensive}.

\begin{figure}
		\centering
		\subfigure[Posterior distribution comparison]{\includegraphics[width=0.47\textwidth]{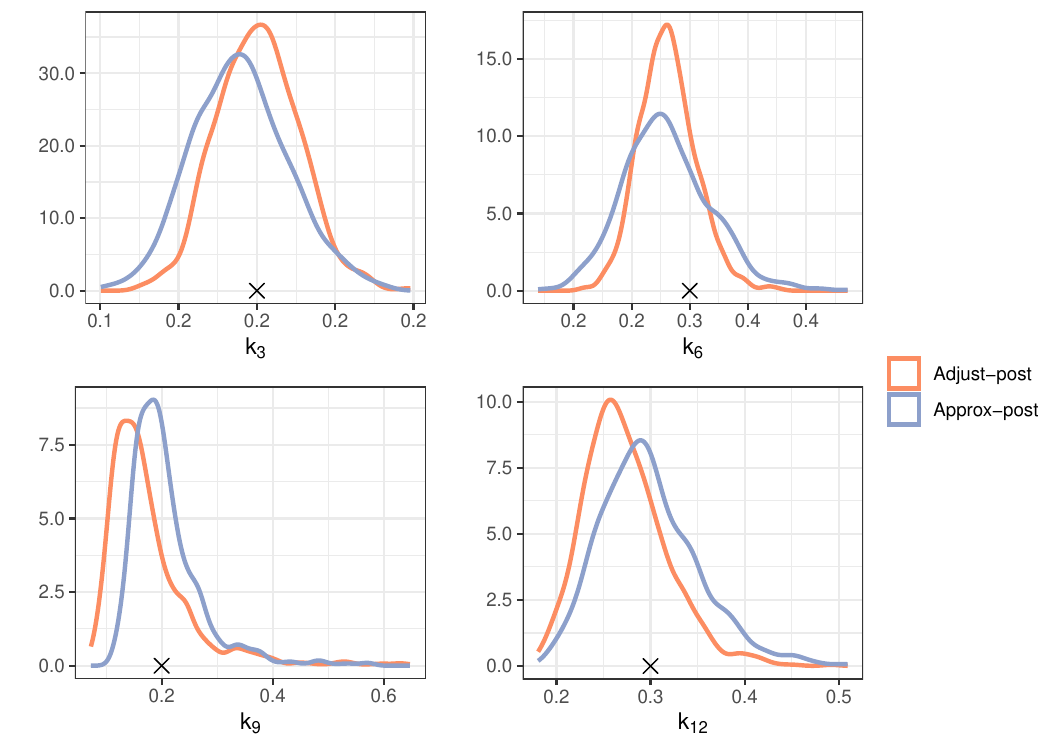}\label{figsub:results_mapk_ekf_post}}
		\subfigure[Calibration diagnostic]{\includegraphics[width=0.47\textwidth]{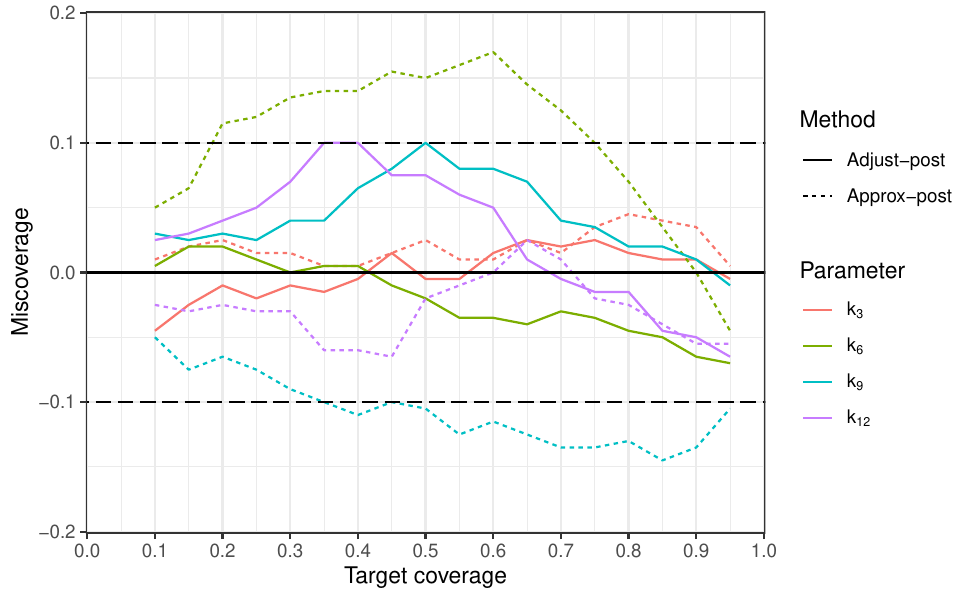}\label{figsub:figsub:results_mapk_ekf_cal}}
		\caption{Comparison of original approximate posterior (Approx-post) and adjusted posteriors (Adjust-post) for reaction network example with EKF likelihood. Plot (a) shows the estimated marginal posterior densities, with true generating parameter value indicated with a cross $(\times)$. Plot (b) is a calibration diagnostic showing the marginal miscoverage for all parameters with $\pm 0.1$ deviation from parity shown with a dotted line.}
		\label{fig:posteriors_mapk_ekf}
\end{figure}

The two-step MAPK model is a stochastic chemical kinetic model for proteins $X$ and $Y$, activated (phosphorylated) proteins $X^a$ and $Y^a$, enzyme $E$, and phosphatase molecules $P_1$ and $P_2$. The two-step MAPK reaction network governs the phosphorylation and dephosphorylation of $X$ and $Y$ by coupled Michaelis–Menten components that can be expressed as 
\begin{equation}\label{eq:mapk}
\begin{aligned}
X + E &\overset{k_1}{\longrightarrow} [XE],& [XE] &\overset{k_2}{\longrightarrow} X + E,& [XE] &\overset{k_3}{\longrightarrow} X^{a}+ E, \\
X^{a} + P_1 &\overset{k_4}{\longrightarrow} [X^{a}P_1],& [X^{a}P_1] &\overset{k_5}{\longrightarrow} X^{a} + P_1,& [X^{a}P_1] &\overset{k_6}{\longrightarrow} X + P_1,\\
X^{a} + Y &\overset{k_7}{\longrightarrow} [X^{a}Y], & [X^{a}Y] &\overset{k_8}{\longrightarrow} X^{a} + Y, & [X^{a}Y] &\overset{k_9}{\longrightarrow} X^{a} + Y^{a},\\
Y^{a} + P_2 &\overset{k_{10}}{\longrightarrow} [Y^{a}P_2], & [Y^{a}P_2] &\overset{k_{11}}{\longrightarrow} Y^{a} + P_2, & [Y^{a}P_2] &\overset{k_{12}}{\longrightarrow} Y + P_2.
\end{aligned}    
\end{equation}
Equation \eqref{eq:mapk} describes the processes of $X$ activating (to $X^a$) by compounding with $E$, and deactivating similarly by $P_1$, whilst $Y$ (to $Y^a$) can be activated by compounding with $X^a$ and deactivated similarly by $P_2$. Hence, there is a cascading effect, or two steps, where $X$ must be activated in order for $Y$ to be activated. The network also allows for the compounds to form without an eventual activation or deactivation taking place. The parameters $k_1,k_2,\ldots,k_{12}$ dictate the rate of the reactions in the system and are the parameters of interest.

Following \citet{warne2022multifidelity} we assume that only the activated proteins are observable with additive normal noise~$(x^a_t, y^a_t) \sim \mathcal{N}((X^{a}_t, Y^{a}_t), \sigma^2 I$), with observations taken at $t \in \{4,8,\ldots,200\}$, known initial counts $E = 94$, $X = 757$, $Y = 567$, $P_1 = P_2 = 32$, and remaining counts zero at $t=0$. The partial observation of proteins leads to parameter unidentifiability, hence we fix $k_1 = k_4 = k_{10} = 10^{-3}$, $k_7 = 10^{-4}$, whilst assigning the following priors~$k_i \sim \mathcal{U}(0,10^{-3})$ iid for $i \in \{2,5,11\}$, $k_j \sim \mathcal{U}(0,1)$ iid for $j \in \{3,6,9,12\}$, and $k_{8} \sim \mathcal{U}(0,10^{-4})$. Using the aforementioned initial conditions and noise $\sigma = 1$ we simulate the observations using the Gillespie's direct method \citep{gillespie1992rigorous}, with parameters $k_1 = k_4 = k_{10} = 10^{-3}$, $k_2 = k_1/120$, $k_3=0.18$, $k_5 = k_4/22$, $k_6 = k_{12} = 0.3$, $k_7 = 10^{-4}$, $k_8 = k_7/110$, $k_9 = 0.2$, and $k_{11} = k_{10}/22$.

The approximate likelihood is defined by two layers of approximations in this example. Firstly, we make a diffusion approximation to the discrete-state continuous-time chemical reaction network, resulting in an SDE.  Then, as in Section~\ref{sec:ex-lotka}, we use the EKF to approximate the SDE's likelihood. We take $\Delta t = 1$ for the intermediary integration times between observations in the EKF, and sample 1000 draws using NUTS-HMC for each approximate posterior. We use $M=200$ calibration samples, unit weights (\ie\ $\alpha=1$), and scale the variance of the approximate posterior (with the observed data) by a factor of $1.5$ to define $\bar\Pi$. We correct the marginal posterior distribution of $(k_3,k_6,k_9,k_{12})$ as the remaining parameters are unidentifiable from the data.

From Figure~\ref{figsub:results_mapk_ekf_post} we can see the correction to the approximate posterior is modest compared to previous examples, indicating the EKF may provide a reasonable approximate likelihood for this model and data set. However, Figure~\ref{figsub:figsub:results_mapk_ekf_cal} provides a strong motivation for calculating and using the adjusted posterior as the over-coverage of $k_6$ and under-coverage of $k_9$ have been corrected to within $\pm 0.1$ for all target coverage values. Despite the unknown properties of the EKF approximation to the MAPK model (and other reaction networks), the calibration diagnostic provides reassurance that this approximation can be reasonable for this data set after the score calibration correction is made.

\section{Discussion} \label{sec:discussion}

In this paper we have presented a new approach for modifying posterior samples based on an approximate model or likelihood to improve inferences with respect to some complex target model.  Our approach does not require any likelihood evaluations of the target model, and only a small number of target model simulations and approximate posterior computations, which are easily parallelizable.  Our approach is particularly suited to applications where the likelihood of the target model is completely intractable, or if the surrogate likelihood is several orders of magnitude faster to evaluate than the target likelihood. We also demonstrated in Section~\ref{sec:ex-mapk} that several layers of approximation could be corrected for.

We focused on correcting inferences from an approximate model, but our approach can also be applied when the inference algorithm is approximate.  For example, we could use our approach to adjust inferences from likelihood-free algorithms or to correct the bias in short MCMC runs.  We plan to investigate this in future research.

We propose a straightforward clipping method when we wish to guarantee finite weights in our importance sampling step. However, a more sophisticated approach to clipping is possible, namely Pareto smoothed importance sampling \citep[PSIS,][]{vehtari2015pareto}. PSIS can be used instead of clipping or unit weights if desired, but we leave investigation of this strategy for future work. 

In general, our method can be used with any proper scoring rule. We concentrated on the energy score because of the ease with which it can be estimated using transformed samples from the approximate distribution. The logarithmic score could also be used by calculating a kernel density estimate from the adjusted posteriors. We found that using $\beta = 1$ for the energy score provided good results in our experiments, but it may be of interest to try $\beta \neq 1$ in future work. 

Using scoring rules that are not strictly proper would alter the objective of Theorem~\ref{th:SBIresult} to recovering certain properties of the target posterior, rather than the entire target posterior itself. Such a choice of scoring rule would simplify the class of sufficiently rich kernels, both asymptotically and non-asymptotically, at the expense of not attempting to recover the full target posterior approximately. This trade-off represents an interesting research direction we are pursuing.

In ongoing work we are considering more flexible transformations when warranted by deficiencies in the approximate posterior. This will be particularly important when the direction of the bias in the approximate model changes in different regions of the parameter space. However, learning more flexible transformations will likely necessitate more calibration samples and data sets, and increased optimization time.  We note that our method with the energy score does not require invertible or differentiable transformations, making it quite flexible compared to most transformation-based inference algorithms which require differentiability.

Another limitation of our approach is that we do not expect our current method to necessarily generate useful corrections when the approximate posterior is very poor. An example where this occurs, to some extent, in a Lotka-Volterra model is provided in Supplement~\ref{sec:lotka-abc}. A poor approximation would likely lead to a family of kernels that is not sufficiently rich. If the correcting transformation is learned in a region of the parameter space far away from the true posterior mass, the calibration data sets are likely to be far away from the observed data, and the transformation may not successfully calibrate the approximate posterior conditional on the observed data. We did provide arguments for when we expect moment-correcting transformations to be asymptotically sufficiently rich and a practical method for detecting insufficiency. However, as already alluded to, a more flexible transformation may be able to calibrate more successfully across a wider set of parameter values (\eg\ the prior or inflated version of the approximate posterior).  We also note that in this paper we assume that the complex target model is correctly specified, and an interesting future direction would consider the case where the target model itself is possibly misspecified.


\acks{JJB was supported by a First Byte Grant from the Centre for Data Science (Queensland University of Technology) and the European Union under the GA 101071601, through the 2023-2029 ERC Synergy grant OCEAN. DJW is supported by an ARC Discovery Early Career Researcher Award (DE250100396). DJN was supported by the Ministry of Education, Singapore, under the Academic Research Fund Tier 2 (MOE-T2EP20123-0009). CD was supported by an Australian Research Council Future Fellowship (FT210100260).}


\appendix

\section{Illustration with Logistic Regression}\label{app:logit}

In this section, we illustrate Bayesian score calibration on a simple worked example to elucidate the framework. Consider binary data $y \in \{0,1\}^n$ and fixed covariate matrix $X \in \Real^{n\times p}$. Suppose we wish to fit a logistic regression with a default prior $\Pi$ for parameter vector $\theta \in \Real^p$ \citep[\eg][]{gelman2008} related to the data by $\mathbb{P}(y = 1) = \frac{1}{1 + \exp(-X\theta)}$. Such a regression would typically involve some type of Monte Carlo sampling algorithm to draw approximate samples from the posterior distribution. Instead, assume that we are interested in using a fast approximation of the posterior distribution that is subsequently calibrated using our method.

Suppose our approximate posterior is generated using the Laplace approximation to a Bayesian logistic regression with flat priors. For a given data set $\sy \in \{0,1\}^n$ this approximation is a multivariate normal distribution $\hat{\Pi}(\cdotmid \sy) = \mathcal{N}(\hat\theta(\sy),\hat\Omega(\sy)^{-1})$ with approximate mean,~$\hat\theta(\sy)$ the logistic regression MLE, and approximate precision,~$\hat\Omega(\sy)$ the Hessian matrix of the negative log-likelihood at the MLE. Both the mean and precision are functions of the simulated data~$\sy$ with fixed covariate matrix $X$. 

We calibrate the Laplace approximation with a correction from the relative moment-correcting transformation family, described in Section~\ref{sec:transform} and denoted here by $\mathcal{F}_{\text{m}}$. This implies that our family of kernels $\kernfam_{\text{m}}$ is described by transformations, $f \in \mathcal{F}_{\text{m}}$, applied to the Laplace approximation for data $\sy$. As such, the best calibrated posterior kernel is selected from
\begin{equation*}
    \kernfam_{\text{m}} = \{\push{f}\hat\Pi: f \in \mathcal{F}_\text{m}\} = \{\mathcal{N}(\hat\theta(\cdot)+b,A\hat\Omega(\cdot)^{-1}A^\top):b\in\Real^p, A\in \mathcal{A}\}.
\end{equation*}
The space $\mathcal{A} \subset \Real^{p\times p}$ is specified in Section~\ref{sec:transform} and ensures the optimal $A$ is unique. 

Several further choices are required to instantiate the Bayesian score calibration framework. 
We use the energy scoring rule with $\beta = 1$, unit stabilizing function, and importance distribution, $\bar{\Pi} = \mathcal{N}(\hat\theta(y),2\hat\Omega(y)^{-1})$, based on the Laplace approximation of the true data, with variance rescaled by $2$.
With this specification, the idealized objective function is
\begin{equation}\label{eq:logisticobj}
    \E_{\theta \sim \bar\Pi} \E_{\sy \sim P(\cdot\mid \theta)}\left\{ r(\theta) \left[\frac{1}{2}\E_{u,u^{\prime} \sim \hat{\Pi}(\cdot\mid y)}\Vert f(u) - f(u^{\prime})\Vert_2  - \E_{u \sim \hat{\Pi}(\cdot\mid y)}\Vert f(u) - \theta \Vert_2 \right]\right\},
\end{equation}
where $r(\theta)$ is the importance weight of prior $\Pi$ to importance distribution $\bar\Pi$, and $\sy \sim P(\cdotmid\theta)$ is the data-generating process implied by the logistic regression. That is, $\mathbb{P}(\sy = 1) = \frac{1}{1 + \exp(-X\theta)}$ and $\mathbb{P}(\sy = 0) = 1-\mathbb{P}(\sy = 1)$ with $X$ fixed by assumption. Maximizing \eqref{eq:logisticobj} for $f\in \mathcal{F}_\text{m}$ yields a correction for the Laplace approximation that is better calibrated to the true posterior distribution (on average, according to the energy score). Finally, using the observed data $y$ and optimal parameters~$(b^\star,A^\star)$ from \eqref{eq:logisticobj}, the calibrated approximate posterior is $\mathcal{N}(\hat\theta(y)+b^\star,A^\star\hat\Omega(y)^{-1}A^{\star\top})$.

To run the Bayesian score calibration algorithm in practice, we construct a Monte Carlo approximation to the objective function \eqref{eq:logisticobj} and find the maximizer $f\in \mathcal{F}_\text{m}$ by optimizing over $b \in \Real^p$ and $A\in\mathcal{A}$. 

\section{Additional Theorems and Proofs}\label{sec:proofs}

\subsection{Proof of Theorem~\ref{th:SBIresult}}\label{pr:SBIresult}
\begin{proof}
Denote the objective function in \eqref{eq:tracoptim} as $E(K)$. Using the Radon-Nikodym derivative $\dd\Pi / \dd \bar\Pi$ we can rewrite this as
\begin{equation*}
    E(K) = \E_{\theta \sim \Pi} \E_{\sy \sim P(\cdotmid \theta)}\left[v(\sy) S(K(\cdotmid \sy),\theta) \right] = \int \Pi(\dd \theta) P(\dd \sy \cd \theta) v(\sy) S(K(\cdotmid \sy),\theta),
\end{equation*}
then substituting $\Pi(\dd \theta) P(\dd \sy \cd \theta) = P(\dd \sy) \Pi(\dd \theta \cd \sy)$ we find that 
\begin{equation*}
    E(K) =  Z_v \E_{\sy \sim Q}\E_{\theta \sim \Pi(\cdotmid \sy)} \left[ S(K(\cdotmid \sy),\theta) \right],
\end{equation*}
where $Q(\dd \sy) = P(\dd \sy)v(\sy)/Z_v$, with $Z_v = P(v)\in (0, \infty)$ by assumption.

Now for maximizing $E(K)$, we can ignore the constant $Z_v$, and find that
\begin{align*}
    K^\star = \argmax_{K \in \kernfam} E(K) &= \argmax_{K \in \kernfam} \E_{\sy \sim Q}\E_{\theta \sim \Pi(\cdotmid \sy)} \left[ S(K(\cdotmid \sy),\theta) \right] \\
    &= \argmax_{K \in \kernfam} \E_{\sy \sim Q}\left[ S(K(\cdotmid \sy),\Pi(\cdotmid \sy)) \right].
\end{align*}
We note that $S(K(\cdotmid \sy),\Pi(\cdotmid \sy))$ is maximized if and only if
$K(\cdotmid \sy) = \Pi(\cdotmid \sy)$ and $K(\cdotmid \sy), \Pi(\cdotmid \sy) \in \probfam$ for fixed $\tilde{y} \in \mathsf{Y}$ since $S$ is a strictly proper scoring rule relative to~$\probfam$. Then under expectation with respect to $Q$, if $\kernfam$ is sufficiently rich, the optima~$K^\star$ must satisfy $K^\star(\cdotmid \sy) = \Pi(\cdotmid \sy)$ almost surely. 
\end{proof}

\subsection{Proof of Theorem~\ref{th:unitweights-as}} \label{pr:unitweights-as}
\begin{proof}
Let $g_z(x)=\frac{\pi(z)}{\bar\pi(z)}\frac{\bar\pi(x)}{\pi(x)}$ and consider $v(\sy_{1:n}) = g(\theta^{\ast}_n)$.  Therefore $w(z,\sy_{1:n}) = g_z(\theta^{\ast}_n)$. Applying the continuous mapping theorem with function $h(x) = g_z(z) - g_z(x) = 1 - g_z(x)$ yields the result.
\end{proof}

\subsection{A Central Limit Theorem for Unit Weights} \label{thpr:unitweights-clt}
\begin{theorem}\label{th:unitweights-clt}
Let $g(x) = \bar \pi(x) / \pi(x)$ for $x \in \Theta$. If there exists an estimator $\theta^{\ast}_n \equiv \theta^{\ast}(\sy_{1:n})$ such that  $\sqrt{n}(\theta^{\ast}_n - z) \distarrow \dNorm(0,\Sigma_z)$ as $n \rightarrow \infty$ when~$\tilde y_i \simiid P(\cdotmid z)$ for $z \in \Theta$, $g(\theta^{\ast}_n) \leq h(\tilde{y}_{1:n})$ a.s. for some integrable function $h$, $g > 0$ a.e., and $\nabla g \neq 0$ a.e., then the error from approximating the weights with $\hat{w} = 1$ as the size of the data $y_{1:n}$ grows satisfies
\begin{align*}
    &\sqrt{n}(\hat{w} - w(\theta,\sy_{1:n})) \distarrow U, \\
    &(U \cd \theta) \sim \dNorm\left(0, \Sigma_\theta^\prime \right), \theta \sim \bar\Pi,
\end{align*}
as $n \rightarrow \infty$ with choice of stabilizing function $v(\sy_{1:n}) = g(\theta^\ast_n)$, where 
\begin{equation*}
    \Sigma_\theta^\prime = \nabla \log g(\theta)^{\top} \Sigma_\theta \nabla\log g(\theta).
\end{equation*} Moreover, $\E(U) = 0$ and the unit weights $\hat{w}$ therefore have asymptotic distribution with variance equal to $\mathrm{var}(U) = \E_{\theta\sim\bar\Pi}(\Sigma_\theta^\prime)$.
\end{theorem}

\begin{proof} Take $\theta = z$ fixed and let $g_z(x)=\frac{\pi(z)}{\bar\pi(z)}\frac{\bar\pi(x)}{\pi(x)}$ for $x \in \Theta$. Consider $v(\sy_{1:n}) = g(\theta^{\ast}_n)$ and therefore $w(z,\sy) = g_z(\theta^{\ast}_n)$.
Using the delta method we can deduce that 
\begin{equation*}
    U_n(z) \equiv  \sqrt{n}(g_z(\theta^{\ast}_n) - g_z(z)) \distarrow U(z), \qquad \text{where} \; U(z) \sim \dNorm\left(0,\nabla g_z(z)^{\top} \Sigma_z \nabla g_z(z) \right),
\end{equation*}
noting that $g_z(\theta^\ast_n) = w(z,\sy)$ and $g_z(z) = 1 = \hat w$.
Now let $\theta \sim \bar\Pi$ on measurable space~$(\Theta, \vartheta)$. For all $A \in \vartheta$, consider
\begin{align*}
    \lim_{n\rightarrow \infty} \Pr(U_n(\theta) \in A) 
    &= \lim_{n\rightarrow \infty} \E_{z \sim \bar \Pi} \Pr(U_n(\theta) \in A \cd \theta = z) \\
    & = \E_{z \sim \bar \Pi} \Pr(U(\theta) \in A \cd \theta = z) \\
     & = \Pr(U(\theta) \in A ),
\end{align*}
by the law of total probability and noting that dominated convergence theorem holds since $0 < g(\theta^{\ast}_n) \leq h(y_{1:n})$ implies that $\vert U_n(z) \vert$ is also dominated. Therefore $U_n(\theta) \distarrow U(\theta)$ as $n \rightarrow \infty$ where $U(\theta) \sim \E_{z \sim \bar \Pi} U(z)$, \ie\ a continuous mixture of Gaussian distributions. Using the continuous mapping theorem we can also state that $U_n \distarrow U$ where $U_n \equiv -U_n(\theta)$ and $U \equiv - U(\theta)$ then noting that $\Sigma_z^\prime \equiv \nabla g_z(z)^{\top} \Sigma_z \nabla g_z(z) = \nabla \log g(z)^{\top} \Sigma_z \nabla \log g(z)$ gives the limiting distribution result. Moreover, we can see $\E(U) = 0$ by the law of total expectation and $\mathrm{var}(U) = \E_{z \sim \bar\Pi}(\Sigma_z^\prime)$ by the law of total variance.
\end{proof}
\begin{remark}
If one wishes to estimate the asymptotic variance of the weight approximation, we have freedom to choose the estimator $\theta^\ast$. If possible we should choose the estimator that results in the smallest asymptotic variance $\mathrm{var}(U)$ or smallest conditional variance $\Sigma_\theta$ if equivalent or more convenient. 
\end{remark}

\begin{remark}
If $\max_{x \in \Theta} g(x) = m < \infty$ then $g(\theta^{\ast}_n) \leq m$ and the dominating condition holds. This indicates that using a distribution $\bar\Pi$ with lighter tails than $\Pi$ is appropriate. Such a statement is surprising as this disagrees with well-established importance sampling guidelines. Moreover, if $\bar\pi(\theta) = \hat\pi(\theta \cd y_{1:n})$ then $g(\theta)$ is the approximate likelihood (ignoring the normalizing constant) and a sufficient condition for the domination is that the approximate likelihood~$g(\theta) = \hat p(y_{1:n} \cd \theta)$ is bounded. This is the case for any approximate likelihood for which a maximum likelihood estimate exists.
\end{remark}

\begin{remark}
The dominating condition can also be enforced by only considering bounded $\Theta$. As such the estimator $\theta^{\ast}_n$ and hence $g(\theta^{\ast}_n)$ will typically be bounded. This is the approach taken by \citet{deistler2022truncated} for sequential neural posterior estimation but no asymptotic justification is given. Our results may be useful in this case, and more generally for this area, but we leave exploration for future research.
\end{remark}

\section{Package and Code Acknowledgments}\label{sec:code}
A \texttt{julia} package with code that can be applied to any approximate model and data-generating process can be found at 
\url{https://github.com/bonStats/BayesScoreCal.jl}, our examples are contained in \url{https://github.com/bonStats/BayesScoreCalExamples.jl}.

Our implementation relies heavily on \texttt{Optim.jl} \citep{Optim.jl-2018} and our examples make use of
\texttt{DifferentialEquations.jl} \citep{DifferentialEquations.jl-2017},
\texttt{Distributions.jl} \citep{JSSv098i16}, and
\texttt{Turing.jl} \citep{ge2018t}.

\section{Idealized Versus Practical Weighting Functions}\label{sec:weightfunctions}

An optimal stabilizing function would necessitate that $w(\theta, \sy) = C$, for some constant $C$, though such a function need not exist. However, considering the properties of a theoretical optimal stabilizing function is useful for our asymptotic results in Section~\ref{sec:uweights}. If there were a deterministic function $g$ perfectly predicting $\theta$ from $\sy$, \ie\ $g(\sy) = \theta$ if $\sy \sim P(\cdotmid \theta)$, then $v(\sy) = \bar\pi[g(\sy)] / \pi[g(\sy)]$ would be the optimal stabilizing function. 

In the absence of such a $g$, we could approximate the stabilizing function by
\begin{equation}\label{eq:stabiliseropt}
    v(\sy) = \frac{\bar\pi[\theta^{\star}(\sy)]}{\pi[\theta^{\star}(\sy)]}, \quad \theta^{\star}(\sy) = \arg\max_{\vartheta \in \Theta} p(\sy \cd \vartheta)\bar\pi(\vartheta),
\end{equation} 
where $\theta^{\star}(\sy)$ is the maximum \textit{a posteriori} (MAP) estimate of $\theta$ given $\sy$. The maximum likelihood estimate could also be used. In the case that $\theta^{\star}(\sy) \approx \theta$ we can deduce that $w(\theta, \sy) \approx C$, though deviations from this may be quite detrimental to the variance of the weights. Unfortunately we do not have access to the likelihood $p(\sy \mid \cdot\,)$, so $\theta^{\star}$ is intractable. The approximate likelihood $\hat p$ is a practical replacement for $p$ but will be likely to further increase the variance of the weights.

As for the importance distribution, a natural way to concentrate $\theta$ about likely values of the posterior given $y$ is to use $\bar\Pi(\cdot) = \hat\Pi(\cdotmid y)$. This generates data sets $\sy$ such that they are consistent with $y$ according to the approximate posterior. The idealized setting, with no Monte Carlo error and an accurate MAP using the approximate likelihood, therefore uses 
\begin{equation*}\label{eq:stabiliseroptapprox}
     v(\sy) = \frac{\bar\pi[\theta^{\diamond}(\sy)]}{\pi[\theta^{\diamond}(\sy)]}, \quad \theta^{\diamond}(\sy) = \arg\max_{\vartheta \in \Theta} \hat p(\sy \cd \vartheta)\bar\pi(\vartheta),
\end{equation*}
with $\bar\Pi(\cdot) =  \hat\Pi(\cdotmid y)$. If $\theta^{\diamond}(\sy)$ is a biased estimator of $\theta$, we could estimate this bias and correct for it to ensure~${w(\theta,\sy) \approx C}$.

In some cases choosing $\bar\Pi(\cdot) =  \hat\Pi(\cdotmid y)$ may be adequate, but it depends crucially on the tail behavior of the ratio~$\pi(\theta)/\bar\pi(\theta)$ and how well the stabilizing function performs. It may be pertinent to artificially increase the variance of the $\bar\Pi$ by transformation or consider an approximation to the chosen distribution with heavier tails. 

The simple countermeasure we consider is to truncate or clip the weights as discussed in Section~\ref{sec:practicalweights}.

\section{Additional Results from Examples}\label{app:tabfigs}
This section contains additional results from the examples in Section~\ref{sec:examples}.  
\begin{longtable}{lrr}
\toprule
 & \multicolumn{2}{c}{Correlation} \\ 
\cmidrule(lr){2-3}
Method & Mean & SD \\ 
\midrule\addlinespace[2.5pt]
Approx-post & $0.00$ & $0.02$ \\ 
Adjust-post (0) & $0.25$ & $0.38$ \\ 
Adjust-post (0.5) & $0.29$ & $0.12$ \\ 
Adjust-post (1) & $0.31$ & $0.12$ \\ 
True-post & $0.41$ & $0.06$ \\ 
\bottomrule
\\
\caption{Summary of empirical correlation between parameter samples of $\rho$ and $D$ over 100 independent data sets for the bivariate OU process example. The posteriors compared are the original approximate posterior (Approx-post) and adjusted posteriors (Adjust-post) with $(\alpha)$ clipping, and the true posterior (True-post). \label{tab:corr_bi_ou}}
\end{longtable}
\clearpage
\begin{longtable}{lrrrr}
\toprule
Method & MSE & Bias & SD & AC$^1$ \\ 
\midrule\addlinespace[2.5pt]
\multicolumn{5}{l}{$\mu$} \\ 
\midrule\addlinespace[2.5pt]
Approx-post & $0.10$ & $0.00$ & $0.22$ & $92$ \\ 
Adjust-post (0) & $0.11$ & $0.00$ & $0.18$ & $73$ \\ 
Adjust-post (0.5) & $0.10$ & $0.01$ & $0.22$ & $90$ \\ 
Adjust-post (1) & $0.10$ & $0.01$ & $0.22$ & $89$ \\ 
True-post & $0.10$ & $0.00$ & $0.22$ & $90$ \\ 
\midrule\addlinespace[2.5pt]
\multicolumn{5}{l}{$D$} \\ 
\midrule\addlinespace[2.5pt]
Approx-post & $2.51$ & $0.11$ & $1.01$ & $84$ \\ 
Adjust-post (0) & $2.99$ & $0.17$ & $0.96$ & $72$ \\ 
Adjust-post (0.5) & $2.73$ & $0.11$ & $1.08$ & $83$ \\ 
Adjust-post (1) & $2.73$ & $0.12$ & $1.08$ & $83$ \\ 
True-post & $2.84$ & $0.15$ & $1.15$ & $87$ \\ 
\midrule\addlinespace[2.5pt]
\multicolumn{5}{l}{$\rho$} \\ 
\midrule\addlinespace[2.5pt]
Approx-post & $0.01$ & $-0.02$ & $0.07$ & $84$ \\ 
Adjust-post (0) & $0.01$ & $-0.01$ & $0.06$ & $73$ \\ 
Adjust-post (0.5) & $0.01$ & $-0.01$ & $0.07$ & $85$ \\ 
Adjust-post (1) & $0.01$ & $-0.01$ & $0.07$ & $86$ \\ 
True-post & $0.01$ & $-0.02$ & $0.08$ & $85$ \\ 
\bottomrule
\\
\caption{Average results for each parameter over 100 independent data sets for the bivariate OU process example. The posteriors compared are the original approximate posterior (Approx-post) and adjusted posteriors (Adjust-post) with $(\alpha)$ clipping. $^1$Achieved coverage with 90\% target.\label{tab:repeated_results_mv_ou}}
\end{longtable}

\bibliography{references}

\clearpage
\renewcommand{\thesection}{S}
\section*{Supplementary Materials}\label{sec:supplement}
The supplementary materials contain three additional examples of the Bayesian score calibration method. 
\subsection{Conjugate Gaussian model}\label{sec:ex-gauss}

\begin{figure}[p]
		\centering
		\includegraphics[scale=0.7]{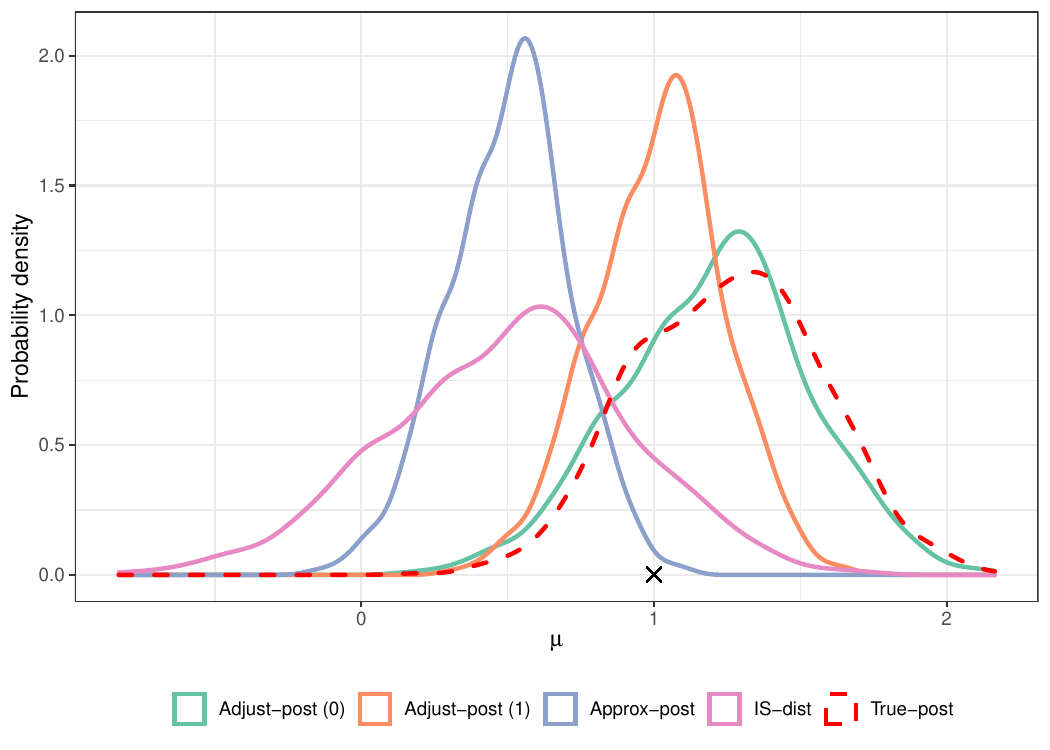}
		\caption{Conjugate Gaussian model univariate densities estimates of approximations to the posterior distribution for a single dataset. The original approximate posterior (Approx-post), importance (IS-dist), and adjusted posteriors (Adjust-post) with $(\alpha)$ clipping are shown with solid lines. The true posterior (True-post) is shown with a dashed line. The true generating parameter value is indicated with a cross $(\times)$.}
		\label{fig:posteriors_normal}
		
\end{figure}

We consider a toy conjugate Gaussian example.  Here the data $y$ are $n=10$ independent samples from a $\dNorm(\mu, \sigma^2)$ distribution with $\sigma^2 = 1$ assumed known and $\mu$ unknown.  Assuming a Gaussian prior $\mu \sim \dNorm(\mu_0, \sigma_0^2)$, the posterior is $(\mu|y) \sim \dNorm(\mu_\mathrm{post}, \sigma_\mathrm{post}^2)$ where
\begin{align*}
    \mu_\mathrm{post} &= \frac{1}{\sigma_0^{-2} + n\sigma^{-2}} \left(  \frac{\mu_0}{\sigma_0^2} + \frac{\sum_{i=1}^ny_i}{\sigma^2} \right) \mbox{  and  }  
    \sigma_\mathrm{post}^2 = \frac{1}{\sigma_0^{-2} + n\sigma^{-2}}.
\end{align*}
We assume this is the target model.  For the approximate model, we introduce random error into the posterior mean and standard deviation
\begin{equation}\label{eq:nnapproxperturb}
    \mu_\mathrm{approx} = \frac{\mu_\mathrm{post} - \mu_\mathrm{error} }{\sigma_\mathrm{error}}  \mbox{  and  }  
    \sigma_\mathrm{approx} = \frac{\sigma_\mathrm{post}}{\sigma_\mathrm{error}},
\end{equation}
where $\mu_\mathrm{error} \sim \dNorm(0.5, 0.025^2)$ and $\sigma_\mathrm{error} \sim \dFNorm(1.5, 0.025^2)$ and $\dFNorm$ denotes the folded-normal distribution. During our simulation the approximate posterior distribution is calculated for each dataset according to \eqref{eq:nnapproxperturb}. The perturbation is random but remains fixed for each dataset. For the stabilizing function we use $v(\sy) = 1$ for simplicity.

We coded this simulation in \texttt{R} \citep{rlang2021} using exact sampling for the true and approximate posteriors. Results based on 100 independent datasets generated from the model with true value $\mu = 1$, and prior parameters $\mu_0 = 0, \sigma_0^2 = 4^2$ are shown in Table \ref{tab:repeated_results_normal}. We truncate the weights from the model at quantiles from the empirical weight distribution. We test truncating the weights for $\alpha = 0$ (no clipping), $0.25, 0.5, 0.9$, and $1$ (uniform weights).  It can be seen that the adjusted approximation (for all $\alpha$) is a marked improvement over the initial approximation, which is heavily biased and has poor coverage.  The estimated posterior distributions based on a single dataset is shown in Figure \ref{fig:posteriors_normal} as an example. We can see that, for this example dataset, the adjusted approximate posteriors are a much better approximation to the true posterior.

\begin{table}[!htp]
    \centering
    \begin{tabular}{lcccc}
    \hline
    Method & MSE & Bias & SD & AC$^1$ \\ 
    \hline
Approx-post & $0.48$ & $-0.64$ & $0.21$ & 0.24 \\ 
Adjust-post (0) & $0.21$ & $-0.16$ & $0.31$ & 0.99 \\ 
Adjust-post (0.25) & $0.15$ & $-0.18$ & $0.26$ & 0.98 \\ 
Adjust-post (0.5) & $0.15$ & $-0.18$ & $0.25$ & 0.98 \\ 
Adjust-post (0.9) & $0.14$ & $-0.18$ & $0.25$ & 0.98 \\ 
Adjust-post (1) & $0.14$ & $-0.18$ & $0.25$ & 0.98 \\ 
True-post & $0.16$ & $0.02$ & $0.31$ & 1.00 \\
    \hline
    \end{tabular}
    \caption{Average results over 100 independent observed datasets for the Gaussian example. The posteriors compared are the original approximate posterior (Approx-post), adjusted posteriors (Adjust-post) with $\alpha$ clipping, and the true posterior (True-post). $^1$Achieved coverage with 90\% target.}
    \label{tab:repeated_results_normal}
\end{table}

\subsection{Fractional ARIMA model} \label{sec:ex-arfima}

Let $\{X_{t}\}_{t=1}^{n}$ be a zero-mean equally spaced time series with stationary covariance function $\kappa(\tau,\theta) = \E(X_{t}X_{t-\tau})$ where $\theta$ is a vector of model parameters.  Here we consider an autoregressive fractionally integrated moving average model (ARFIMA) model for $\{X_{t}\}_{t=1}^{n}$, described by the polynomial lag operator equation as
\begin{align*}
\phi(L)(1-L)^d X_t = \vartheta(L)\epsilon_t,
\end{align*}
where $\epsilon_t \sim \dNorm(0,\sigma^2)$, $L$ is the lag operator, $\phi(z) = 1 - \sum_{i=1}^p \phi_i z_i^i$ and $\vartheta(z) = 1 + \sum_{i=1}^q \vartheta_i z_i^i$.  We denote the observed realized time series as $y = (y_1,y_2,\ldots,y_n)^\top$ where $n$ is the number of observations.  Here we consider an ARFIMA($p,d,q$) model where $p = 2$, $q=1$, and the observed data is simulated with the true parameter $\theta = (\phi_1, \phi_2, \vartheta_1, d)^\top = (0.45, 0.1, -0.4, 0.4)^\top$ with $n = 15,000$.  As in \citet{bon2021accelerating}, we impose stationarity conditions by transforming the polynomial coefficients of $\phi(z)$ and $\vartheta(z)$ to partial autocorrelations \citep{barndorff1973parametrization} taking values on $[-1,1]^p$ and $[-1,1]^q$ respectively, to which we assign a uniform prior. We apply an inverse hyperbolic tangent transform to map the ARMA parameters to the real line to facilitate posterior sampling.  The fractional parameter $d$ has bounds $(-0.5, 0.5)$, we sample over the transformed parameter $\tilde{d} = \mbox{tanh}^{-1}(2d)$, and assume that $\tilde{d} \sim \dNorm(0,1)$ \emph{a priori}. 

The likelihood function of the ARFIMA model for large $n$ is computationally intensive.  As in, for example, \citet{salomone2020spectral} and \citet{bon2021accelerating}, we use the Whittle likelihood \citep{whittle1953estimation} as the approximate likelihood to form the approximate posterior.  Transforming both the data and the covariance function to the frequency domain enables us to construct the Whittle likelihood with these elements rather than using the time domain as inputs. The Fourier transform of the model's covariance function, or the spectral density $f_{\theta}(\omega)$, is
$$
f_{\theta}(\omega) = \frac{1}{2\pi} \sum_{\tau = -\infty}^{\infty}\kappa(\tau,\theta)\exp(-i \omega\tau),
$$
where the angular frequency $\omega \in (-\pi, \pi]$. Whereas the discrete Fourier transform (DFT) of the time series data is defined as
\begin{equation*}
J(\omega_{k}) = \frac{1}{\sqrt{2\pi}}\sum_{t = 1}^{n}X_{t}\exp(-i \omega_{k}t),\quad  \omega_{k} = \frac{2\pi k}{n},
\end{equation*}
using the Fourier frequencies $\{\omega_{k}: k = -\lceil n/2 \rceil+1,\ldots, \lfloor n/2 \rfloor\}$. Using the DFT we can calculate the periodogram, which is an estimate of the spectral density based on the data:
$$
\mathcal{I}(\omega_{k}) = \frac{|J(\omega_{k})|^2}{n}.
$$
Then the Whittle log-likelihood \citep{whittle1953estimation} can be defined as
$$
\ell_{\text{whittle}}(\theta) = -\sum_{k = -\lceil n/2 \rceil+1}^{\lfloor n/2 \rfloor}\left(\log f_{\theta}(\omega_{k}) + \frac{\mathcal{I}(\omega_{k})}{f_{\theta}(\omega_{k})}\right).
$$
In practice the summation over the Fourier frequencies, $\omega_{k}$, need only be evaluated on around half of the values due to symmetry about $\omega_{0}=0$ and since $f_{\theta}(\omega_{0}) = 0$ for centred data.

The periodogram can be calculated in $\mathcal{O}(n \log n)$ time, and only needs to be calculated once per dataset. After dispersing this cost, the cost of each subsequent likelihood evaluation is $\mathcal{O}(n)$, compared to the usual likelihood cost for time series (with dense precision matrix) which is $\mathcal{O}(n^2)$.

Since this example is more computationally intensive, we do not repeat the whole process 100 times in our simulation study.  Instead, we fix the observed data and base the repeated dataset results on the 100 calibration datasets (generated from the 100 calibration parameter values) that are produced in a single run of the process.  However, we do not validate based on the datasets used in the calibration step, but rather generate 100 fresh datasets from the calibration parameter values. As such, we do not provide a comparison to the true posterior in Table~\ref{tab:repeated_results_whittle} as it is fixed and has a high computational cost to sample from. We consider univariate moment-correcting transformations, i.e. $A$ in (8) 
is diagonal, since the covariance structure is well approximated by the Whittle likelihood in this example.  We also choose the stabilizing function to be $v(\sy) = 1$ for simplicity.

To generate samples from the approximate and true posterior distributions we use a sequential Monte Carlo sampler \citep{del2006sequential}. In particular, we use likelihood annealing with adaptive temperatures \citep{jasra2011inference,beskos2016convergence} and a Metropolis-Hastings mutation kernel with a multivariate Gaussian proposal. The covariance matrix is learned adaptively as in \citet{chopin2002sequential}. The simulation is coded in \texttt{R} \citep{rlang2021}.

\begin{figure}
		\centering
		\includegraphics[width=\textwidth]{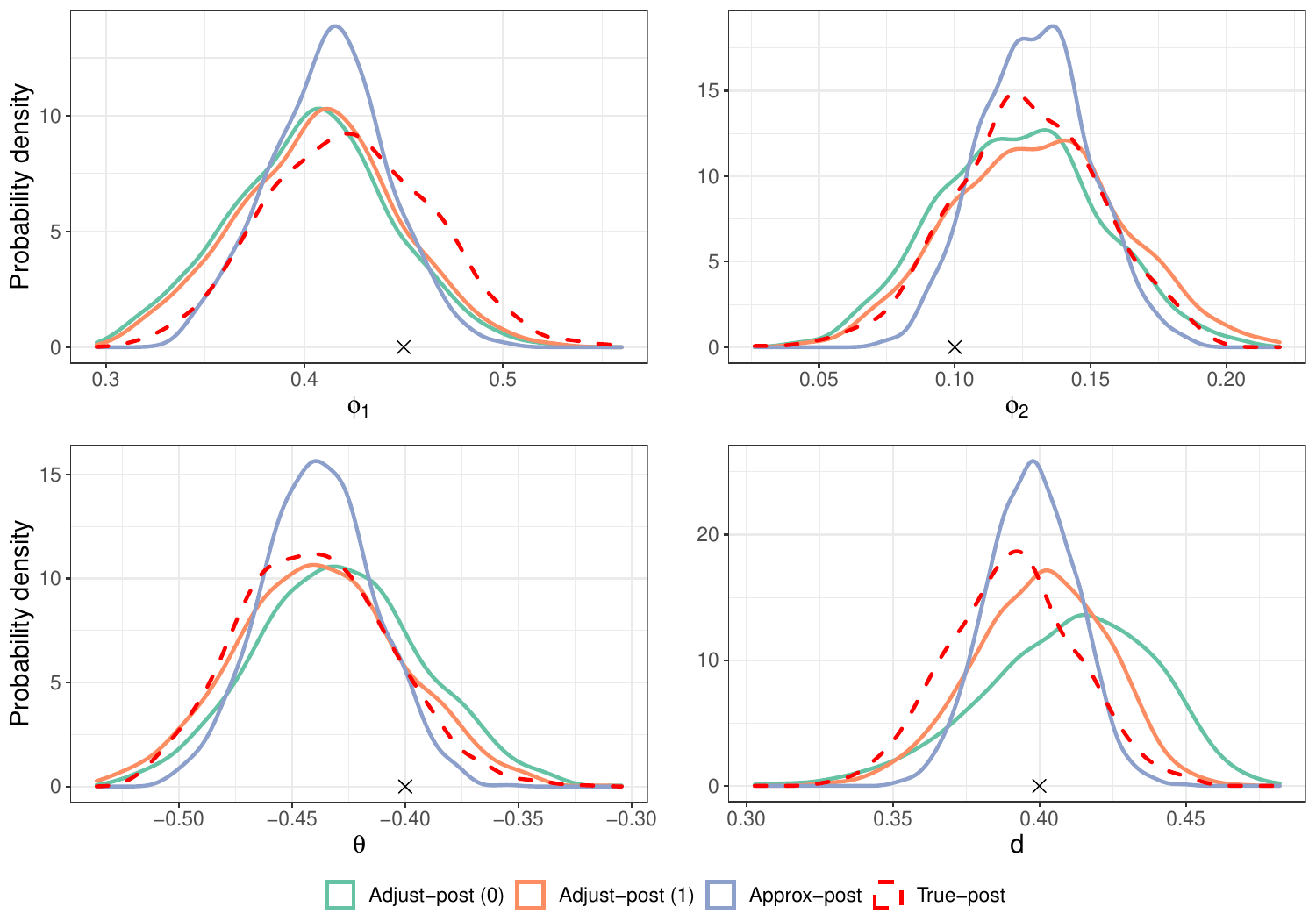}
		\caption{Estimated univariate posterior distributions for the Whittle likelihood example.  Distributions shown are the original approximate posterior (Approx-post) and adjusted posteriors (Adjust-post) with $(\alpha)$ clipping. The true posterior (True-post) is shown with a dashed line. The true generating parameter value is indicated with a cross $(\times)$.}
		\label{fig:posteriors_whittle}
\end{figure}

The repeated run results for the parameters are shown in Table~\ref{tab:repeated_results_whittle}.  It is evident that the Whittle approximation performs well in terms of estimating the location of the posterior, but the estimated posterior standard deviation is slightly too small, which leads to some undercoverage.  The adjusted posteriors inflate the variance and obtain more accurate coverage of the calibration parameters.  An example adjustment for the true dataset is shown in Figure \ref{fig:posteriors_whittle}, which shows that the adjustment inflates the approximate posterior variance. We also compute the calibration diagnostic to confirm the method is performing appropriately on the original data. Figure~\ref{fig:cal_whittle} shows that the achieved coverage is close to the target coverage, across a range of targets, hence the method is performing well.

\begin{longtable}{lrrrr}
\toprule
Method & MSE & Bias & SD & AC$^1$ \\ 
\midrule
\multicolumn{1}{l}{$\phi_1$} \\
\midrule
Approx-post & $0.004$ & $0.003$ & $0.031$ & 67 \\ 
Adjust-post (0) & $0.005$ & $-0.007$ & $0.041$ & 83 \\ 
Adjust-post (0.5) & $0.004$ & $-0.003$ & $0.041$ & 83 \\ 
Adjust-post (1) & $0.005$ & $-0.003$ & $0.041$ & 83 \\ 
\midrule
\multicolumn{1}{l}{$\phi_2$} \\
\midrule
Approx-post & $0.002$ & $-0.001$ & $0.021$ & 74 \\ 
Adjust-post (0) & $0.002$ & $-0.007$ & $0.031$ & 86 \\ 
Adjust-post (0.5) & $0.002$ & $0.001$ & $0.031$ & 89 \\ 
Adjust-post (1) & $0.002$ & $0.000$ & $0.032$ & 89 \\ 
\midrule
\multicolumn{1}{l}{$\vartheta_1$} \\ 
\midrule
Approx-post & $0.002$ & $0.000$ & $0.025$ & 73 \\ 
Adjust-post (0) & $0.003$ & $0.009$ & $0.037$ & 87 \\ 
Adjust-post (0.5) & $0.003$ & $0.000$ & $0.036$ & 88 \\ 
Adjust-post (1) & $0.003$ & $0.000$ & $0.037$ & 88 \\ 
\midrule
\multicolumn{1}{l}{$d$} \\ 
\midrule
Approx-post & $0.001$ & $-0.004$ & $0.017$ & 74 \\ 
Adjust-post (0) & $0.002$ & $0.008$ & $0.032$ & 96 \\ 
Adjust-post (0.5) & $0.001$ & $0.000$ & $0.024$ & 89 \\ 
Adjust-post (1) & $0.001$ & $-0.001$ & $0.025$ & 90 \\ 
\bottomrule
\\
\caption{Average results for each parameter over 100 independent calibration datasets (fixed observation dataset) for the Whittle example. The posteriors compared are the original approximate posterior (Approx-post) and adjusted posteriors (Adjust-post) with $(\alpha)$ clipping. $^1$Achieved coverage with 90\% target.\label{tab:repeated_results_whittle}}
\end{longtable}

\begin{figure}
		\centering
		\includegraphics[width=0.5\textwidth]{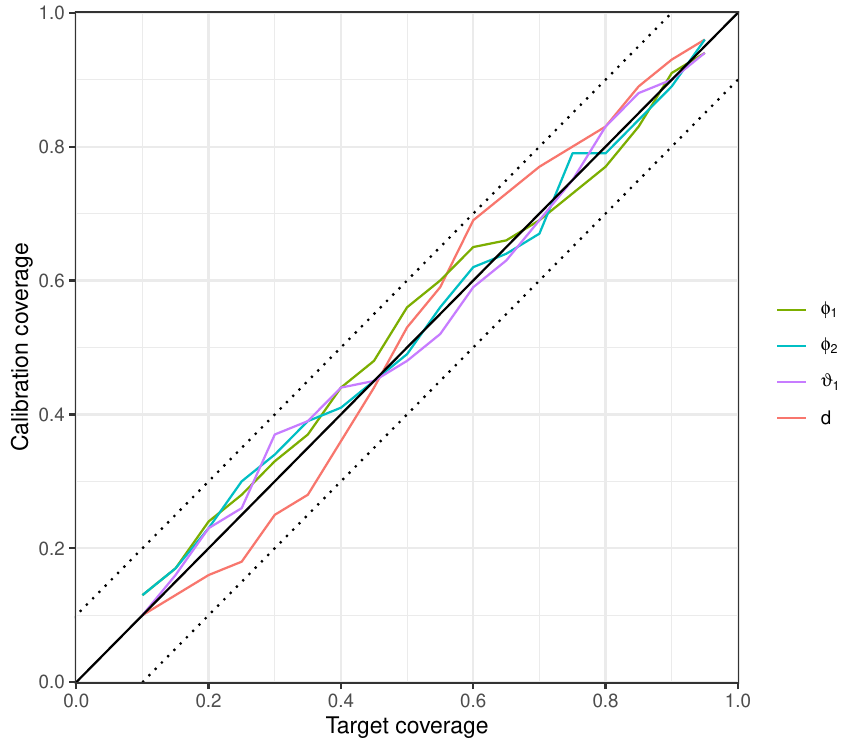}
		\caption{Calibration checks for all parameters in the Whittle likelihood example for $\alpha = 1$ with $\pm 0.1$ deviation from parity shown with a dotted line.}
		\label{fig:cal_whittle}
\end{figure}

\subsection{Lotka-Volterra model with ABC-like posterior}\label{sec:lotka-abc}

We also interested in testing how well our method can perform in a situation where there are no guarantees about the approximate model being employed. We test an alternative approximate posterior for the Lotka-Volterra SDE in Section~3.2 
once\footnote{See \url{https://github.com/TuringLang/Turing.jl/issues/2216} for discussion, the SDE example has now been removed from the tutorial but is available here \url{https://web.archive.org/web/20210419133415/https://turing.ml/dev/tutorials/10-bayesiandiffeq/}.} appearing in the \texttt{Turing.jl} tutorials \citep{ge2018t}. It used an unreferenced method for inference on the parameters of an SDE similar to an ABC posterior but with unknown kernel bandwidth. Despite its unknown inferential properties, we can correct the approximation using Bayesian score calibration and assess the correction using the calibration diagnostics.  For the approximate model we use the noisy quasi-likelihood 
\begin{equation*}
    l(\beta_{1:4}, \tau \cd x_{1:n}, y_{1:n}) = \tau^{2n}\exp\left(-\frac{\tau^2}{2} \sum_{i=1}^{n}\left[(x_i^\prime - x_i)^{2} + (y_i^\prime - y_i)^{2}\right] \right),
\end{equation*}
where $\{(x_i^\prime,y_i^\prime)\}_{i=1}^{n}$ are simulated conditional on the $\beta_{1:4}$ using a rough approximation to the SDE (14). 
In particular, we use the Euler-Maruyama method with $\Delta t = 0.01$. For priors we use $\beta_i \simiid \dUnif(0.1,5)$ for $i \in \{1,2,3,4\}$ and~${\tau \sim \dGam(2,3)}$.  The quasi-likelihood is reminiscent of an approximate likelihood used in ABC.  In particular, a Gaussian kernel is used to compare observed and simulated data, where $\tau$ plays a similar role to the tolerance in ABC.  Usually in ABC, the tolerance is chosen to be small and fixed, but here we take it as random.  However, as shown in \citet{bortot2007inference}, for example, using a random tolerance can enhance mixing of the MCMC chain.  The resulting approximate posterior therefore has two levels of approximation; (i) an ABC-like posterior which uses (ii) a coarse approximate simulator rather than the true data generating process (which would require relatively more computation).
\begin{figure}
		\centering
		\includegraphics[width=0.9\textwidth]{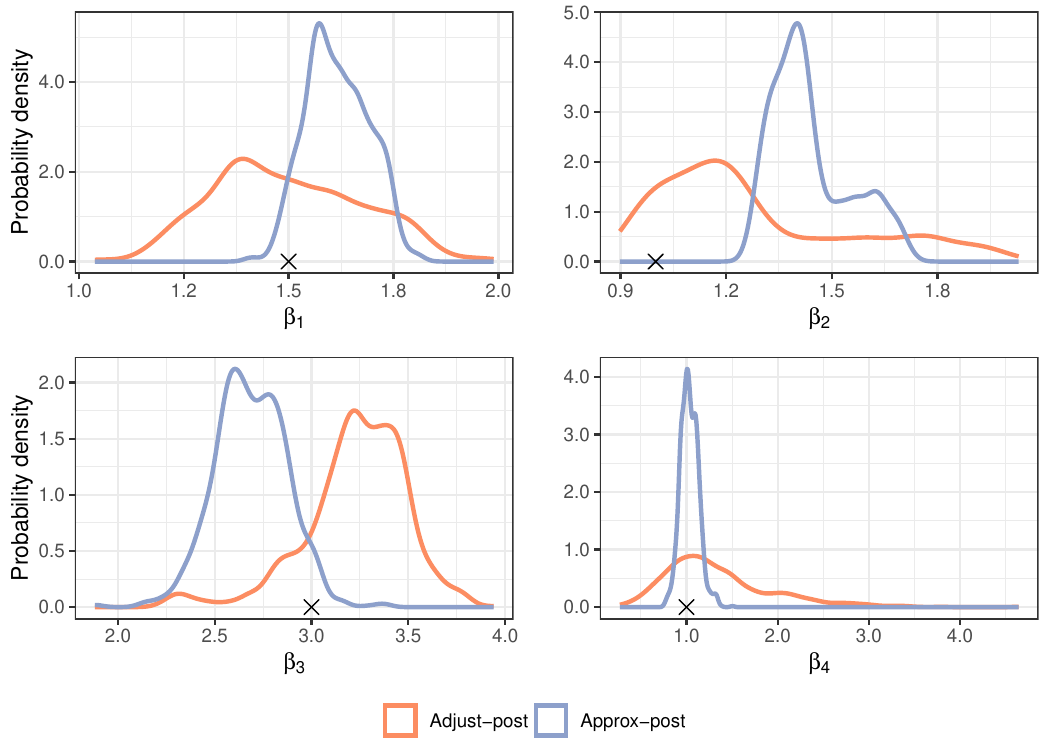}
		\caption{Estimated marginal posterior distributions for the Lotka-Volterra example.  Distributions shown are the original approximate posterior (Approx-post) and adjusted posteriors (Adjust-post) with $\alpha = 1$ clipping.  The true generating parameter value is indicated with a cross $(\times)$.}
		\label{fig:posteriors_lotka}
\end{figure}
For the calibration procedure we use $M = 200$ calibration datasets and use the unit weight approximation. We draw samples from the approximate posterior using the NUTS Hamiltonian Monte Carlo sampler with 0.25 target acceptance rate since the gradient is noisy. Finally, we use a multivariate moment-correcting transformation with dimension $d=4$ and add a squared penalty term to all parameters of the lower-diagonal scaling matrix $L$. We shrink the diagonal elements of $L$ to one, and off-diagonals elements to zero with rate $\lambda = 0.05$.

The results show that the adjusted posterior in this example does significantly better than the approximate posterior. Figure~\ref{fig:posteriors_lotka} shows the marginal posteriors have much better coverage of the true parameter values and generally increase the variance of the approximate posterior to reflect that the approximation is too precise. Significant bias in parameter~$\beta_2$ is also corrected. The target 90\% coverage estimate from the calibration diagnostic was $(0.32, 0.31, 0.29, 0.37)$ for the approximate posterior, and $(0.68, 0.79, 0.69, 0.97)$ for the adjusted posterior for $(\beta_1,\beta_2, \beta_3, \beta_4)$ respectively.

Despite overall positive results, the calibration diagnostic shows that caution is required when using this adjusted approximate posterior. Though we see the adjusted posterior has significantly better achieved coverage than the approximate posterior in Figure~\ref{fig:cal_lotka}, it still lags the nominal target coverage for parameters $\beta_1$ and $\beta_3$. Hence a conservative approach should be taken to constructing credible regions for those parameters. The calibration diagnostic is a useful byproduct of our method, alerting users to potential problems.

\begin{figure}
		\centering
		\includegraphics[width=0.9\textwidth]{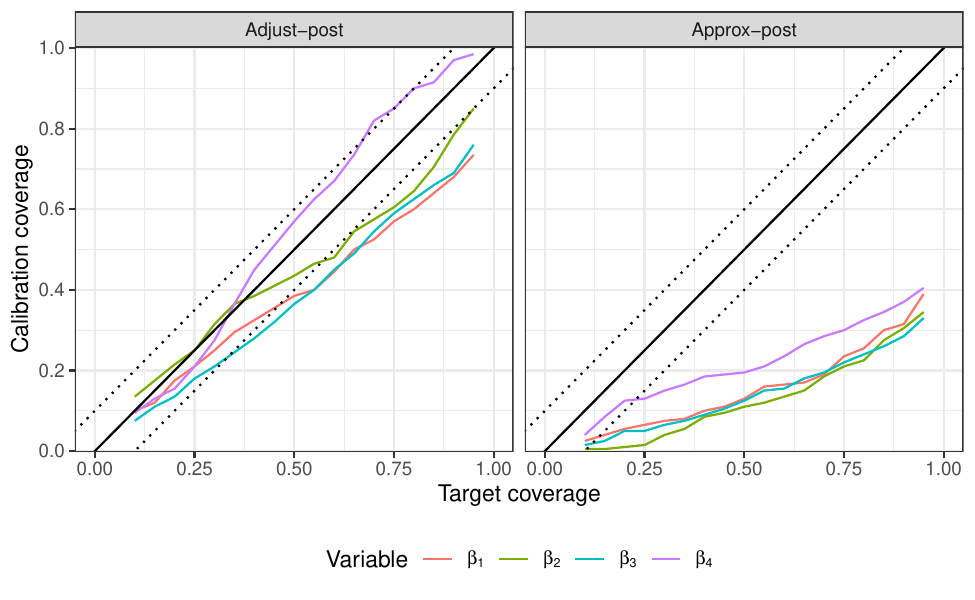}
		\caption{Calibration checks for all parameters in the Lotka-Volterra (ABC-like posterior) example for $\alpha = 1$ with $\pm 0.1$ deviation from parity shown with a dotted line.}
		\label{fig:cal_lotka}
\end{figure}

We also note a potential over-correction in the bias of $\beta_3$ seen in Figure~\ref{fig:posteriors_lotka}. Investigating this further we find that the calibration samples that were simulated were in the range of $\beta_3 \in [1.5, 2.9]$ which excluded the true parameter value. Hence the transformation learned may have been skewed to this region. Increasing the scale of the importance distribution that generates the calibration samples may help this. Using a more sophisticated transformation or performing the calibration sequentially, as in \citet{pacchiardi2022likelihood}, may also be possible in future work. The ranges covered by the importance distribution for the other parameters were $\beta_1 \in [1.5,  2.4], \beta_2 \in [0.9, 2.7]$ and $\beta_4 \in [0.6, 1.1]$. A different importance distribution should be used if one suspects these ranges do not adequately cover the likely values of the true posterior.

\end{document}